\documentclass[11pt]{article}
\usepackage[utf8]{inputenc}

\usepackage{graphicx}
\usepackage[dvipsnames]{xcolor}
\usepackage{amsmath,amsthm,amssymb}
\usepackage[margin=1in]{geometry}
\usepackage[pdftex,colorlinks=true,linkcolor=blue,citecolor=blue,urlcolor=black]{hyperref}
\usepackage{comment}
\usepackage[normalem]{ulem}
\usepackage{cite}
\usepackage{multirow}

\newcommand{\E}{\mathbb{E}}
\newcommand{\F}{\mathbb{F}}

\def\a{{\bf a}}
\def\b{{\bf b}}
\def\c{{\bf c}}
\def\e{{\bf e}}

\def\x{{\bf x}}
\def\y{{\bf y}}
\def\z{{\bf z}}
\def\s{{\bf s}}
\def\q{{\bf q}}
\def\n{{\bf n}}

\def\v{{\bf v}}
\def\w{{\bf w}}
\def\E{{\mathbb E}}
\def\T{{\mathcal T}}

\newcommand{\ket}[1]{| #1 \rangle}
\newcommand{\bra}[1]{\langle #1|}

\newcommand{\be}{\begin{equation}}
\newcommand{\ee}{\end{equation}}
\newcommand{\bea}{\begin{eqnarray}}
\newcommand{\eea}{\end{eqnarray}}
\newcommand{\bes}{\begin{equation*}}
\newcommand{\ees}{\end{equation*}}
\newcommand{\beas}{\begin{eqnarray*}}
\newcommand{\eeas}{\end{eqnarray*}}

\makeatletter
\newtheorem*{rep@theorem}{\rep@title}
\newcommand{\newreptheorem}[2]{%
\newenvironment{rep#1}[1]{%
 \def\rep@title{#2 \ref{##1} (restated)}%
 \begin{rep@theorem}}%
 {\end{rep@theorem}}}
\makeatother 

\allowdisplaybreaks

\newtheorem{thm}{Theorem}
\newtheorem*{thm*}{Theorem}

\newtheorem{lem}[thm]{Lemma}
\newtheorem*{lem*}{Lemma}

\newtheorem{prop}[thm]{Proposition}
\newtheorem{defn}[thm]{Definition}
\newtheorem{rem}[thm]{Remark}

\newreptheorem{thm}{Theorem}
\newreptheorem{lem}{Lemma}
\newreptheorem{cor}{Corollary}

\usepackage{pgfplots}

\usepackage[title]{appendix}
\usepackage{authblk}

\title{ Faster quantum-inspired algorithms for solving linear systems}

\author[1]{Changpeng Shao\thanks{changpeng.shao@bristol.ac.uk}}
\author[1,2]{Ashley Montanaro\thanks{ashley.montanaro@bristol.ac.uk}}
\affil[1]{School of Mathematics, University of Bristol, UK}
\affil[2]{Phasecraft Ltd. UK}

\date{\today}

\begin{document}

\maketitle


\begin{abstract}

We establish an improved classical algorithm for solving linear systems in a model analogous to the QRAM that is used by quantum linear solvers. Precisely, for the linear system $A\x = \b$, we show that there is a classical algorithm that outputs a data structure for $\x$ allowing sampling and querying to the entries, where $\x$ is such that $\|\x - A^{+}\b\|\leq \epsilon \|A^{+}\b\|$. This output can be viewed as a classical analogue to the output of quantum linear solvers.
The complexity of our algorithm is  $\widetilde{O}(\kappa_F^4 \kappa^2/\epsilon^2 )$, where $\kappa_F = \|A\|_F\|A^{+}\|$ and $\kappa = \|A\|\|A^{+}\|$. This improves the previous best algorithm [Gily{\'e}n, Song and Tang,  arXiv:2009.07268] of complexity $\widetilde{O}(\kappa_F^6 \kappa^6/\epsilon^4)$.
Our algorithm is based on the randomized Kaczmarz method, which is a particular case of stochastic gradient descent. We also find that when $A$ is row sparse, this method already returns an approximate solution $\x$ in time $\widetilde{O}(\kappa_F^2)$, while the best quantum algorithm known returns $\ket{\x}$ in time $\widetilde{O}(\kappa_F)$ when $A$ is stored in the QRAM data structure. As a result, assuming access to QRAM and if $A$ is row sparse, the speedup based on current quantum algorithms is quadratic.

\end{abstract}


\section{Introduction}

Since the discovery of the Harrow-Hassidim-Lloyd algorithm \cite{harrow2009quantum}
for solving linear systems, many quantum algorithms for solving machine learning problems were proposed, e.g. \cite{kerenidis2017quantum,kerenidis2019qmeans,kerenidis2020quantum,lloyd2014quantum,rebentrost2014quantum,biamonte2017quantum}. Most of them claimed that quantum computers could achieve exponential speedups over classical algorithms. However, as shown by Tang \cite{tang2019quantum} and some subsequent works \cite{chia2020sampling,gilyen2018quantum,jethwani2019quantum,tang2018quantum,chia2019quantum,gilyen2020improved,chepurko2020quantum,chia2018quantum2} (currently known as quantum-inspired algorithms), for many previously discovered quantum machine learning algorithms, the quantum computer actually only achieves a polynomial speedup in the dimension.
Tang's work makes us reconsider the advantages of quantum computers in machine learning. For instance, one may wonder what is the best possible separation between quantum and classical computers for solving a given machine learning problem?

Here we focus on the solving of linear systems, which is a fundamental problem in machine learning. 
For a linear system $A\x=\b$, as far as we know, the currently best quantum algorithm seems to be the one building on the technique of block-encoding.
For simplicity, assume that the singular values of $A$ lie in $[1/\kappa,1]$, where  $\kappa$ is the condition number of $A$.
Let $\alpha \in \mathbb{R}^+$, if $A/\alpha$ is efficiently encoded as a block of a unitary, the complexity to obtain $\ket{A^{+}\b}$ up to error $\epsilon$ is $\widetilde{O}(\alpha\kappa)$~\cite{chakraborty_et_al:LIPIcs:2019:10609}, where the $\widetilde{O}$ notation hides logarithmic factors in all parameters (including $\kappa, \epsilon$ and the size of $A$). For example, when we have the QRAM data structure for $A$, we can take $\alpha= \|A\|_F$.\footnote{In fact,  $\alpha = \mu(A) $ for some $\mu(A) \leq \|A\|_F$ (see \cite{gilyen2019quantum}). The definition of $\mu(A)$ is quite complicated. In this paper, for convenience of the comparison we just use $\|A\|_F$. }
When $A$ is $s$-sparse and given in the sparse-access input model, we can take $\alpha = s\|A\|_{\max}$. 
In the classical case, the quantum-inspired algorithm given in \cite{chia2020sampling} costs $\widetilde{O}(\kappa_F^6 \kappa^{22}/\epsilon^6)$ to output a classical analogue to $\ket{A^{+}\b}$, where $\kappa_F = \|A\|_F \|A^{+}\|$,\footnote{In the literature, $\kappa_F$ does not have a standard name. It is named the scaled condition number in \cite{strohmer2009randomized,gower2015randomized}.} and $\epsilon$ is the accuracy. This result was  recently improved to $\widetilde{O}(\kappa_F^6 \kappa^6/\epsilon^4)$~\cite{gilyen2020improved}. Although the quantum-inspired algorithms are polylog in the dimension, the heavy dependence on $\kappa_F, \kappa, \epsilon$ may make them hard to use in practice \cite{arrazola2019quantum}. 
So it is necessary to improve the efficiency of quantum-inspired algorithms. Meanwhile, this can help us better understand the quantum advantages in certain problems.

QRAM is a widely used data structure in many quantum machine learning algorithms (e.g.~\cite{kerenidis2017quantum,kerenidis2019qmeans,kerenidis2020quantum}). It allows us efficiently to  prepare a quantum state corresponding to the data and to perform quantum linear algebraic operations.
The quantum-inspired algorithms usually assume a data structure which is a classical analogue to the QRAM. For a vector, this data structure allows us to do the sampling and query operations to the entries. The sampling operation can be viewed as measurements of the quantum state. 

As for the problem of solving linear systems,
in the QRAM model, the input of quantum algorithms is the QRAM data structures for $A$ and $\b$, and the output is $\ket{\x}$ such that $\|\ket{\x} - \ket{A^{+}\b}\|\leq \epsilon$. 
In the classical setting, the input of quantum-inspired algorithms is the classical analogue of QRAM data structures for $A$ and $\b$, and the output is a data structure for $\x$ such that $\|\x - A^{+}\b\| \leq \epsilon \|A^{+}\b\|$. This makes the comparison of quantum/quantum-inspired algorithms reasonable. There are also some other quantum models, such as the sparse-access input model. This allows efficient quantum algorithms for solving sparse linear systems \cite{childs2017quantum}. However, because of the similarity between QRAM and the model used in quantum-inspired algorithms, in this paper we will mainly focus on the QRAM model for quantum linear solvers.

Randomized iterative methods have received a lot of attention recently in solving linear systems because of their connection to the stochastic gradient descent method. 
The randomized Kaczmarz method \cite{strohmer2009randomized} is a typical example.  It was first proposed by Kaczmarz in 1937 \cite{karczmarz1937angenaherte}, and rediscovered by Gordon,  Bender and Herman in 1970 \cite{gordon1970algebraic}. After that, there appear many generalizations \cite{gower2015randomized,moorman2020randomized,necoara2019faster,richtarik2020stochastic}. Kaczmarz method is a particular case of the stochastic gradient descent algorithm for solving linear systems. In its simplest version, at each step of iteration, it randomly chooses a row of the matrix with probability corresponding to the norm of that 
row, then finds the closest solution that satisfies this linear constraint. Geometrically, the next point is the orthogonal projection of the previous one onto the hyperplane defined by the chosen linear constraint. The simple iterative structure and clear geometric meaning greatly simplifies the analysis of this method and makes it popular in applications like computed tomography \cite{censor1997parallel}. 
Usually, the randomized iterative methods need to do some importance sampling according to the norms of the rows. Sampling access is made available by the data structure for quantum-inspired algorithms and the QRAM for quantum linear solvers. So in this paper, for the fairness of comparison we also assume  a similar data structure for the randomized iterative methods.

\subsection{Main results}

In this paper, we propose two quantum-inspired algorithms for solving linear systems. One is based on the randomized Kaczmarz method \cite{strohmer2009randomized}, and the other one is based on the randomized coordinate descent method \cite{leventhal2010randomized}. The second algorithm only works for symmetric positive-definite (SPD) matrices. Our results are summarized in Table \ref{tab:comparison}.

{\renewcommand
\arraystretch{1.3}
\begin{table}[h]
\centering
\begin{tabular}{|c|c|c|c|c|} 
 \hline
  & Complexity & Reference & Assumptions  \\ \hline
 Quantum algorithm  & $\widetilde{O}(\kappa_F)$ & \cite{chakraborty_et_al:LIPIcs:2019:10609}, 2018 & \\ \hline
 \multirow{2}{*}{Randomized classical algorithm}    & $\widetilde{O}( s\kappa_F^2  )$ & \cite{strohmer2009randomized}, 2009 & row sparse  \\ 
  & $\widetilde{O}(s {{\rm Tr}(A) \|A^{+}\| })$ &
  \cite{leventhal2010randomized}, 2010 & sparse, SPD \\  \hline
 \multirow{3}{*}{Quantum-inspired algorithm}  & $\widetilde{O}(\kappa_F^6 \kappa^6/\epsilon^4)$ & \cite{gilyen2020improved}, 2020 & \\ 
  & $\widetilde{O}(\kappa_F^4 \kappa^2/\epsilon^2 )$ & Theorem \ref{thm:improvements} &  \\ 
  & $\widetilde{O} ( {{\rm Tr}(A)^2\|A^{+}\|^2\kappa}/{\epsilon^2} )$ &
  Theorem \ref{thm:spd} & SPD \\
 \hline
\end{tabular}
\caption{Comparison of different algorithms for solving the linear system $A\x=\b$, where $\kappa_F=\|A\|_F \|A^{+}\|, \kappa = \|A\| \|A^{+}\|$ and $s$ is the row sparsity. SPD = symmetric positive definite.}
\label{tab:comparison}
\end{table}
}

In the table, the quantum algorithm is based on the QRAM model. The randomized classical and quantum-inspired algorithms assume a data structure that is a classical analogue to the QRAM (see Section \ref{sec:Preliminaries of quantum-inspired algorithms}). So all those algorithms are based on similar data structures, which makes the comparison in the table fair enough. 

\begin{itemize}
    \item When $A$ is dense, the complexity of our algorithm is $\widetilde{O}(\kappa_F^4 \kappa^2/\epsilon^2 )$. If $A$ is additionally SPD, the complexity becomes $\widetilde{O} ( {{\rm Tr}(A)^2\|A^{+}\|^2\kappa}/{\epsilon^2} )$. This reduces the complexity of the previous quantum-inspired algorithm given by Gily\'{e}n et al. \cite{gilyen2020improved} by at least a 4th power in terms of the dependence on $\kappa$, improves the dependence on $\kappa_F$ from 6th to 4th, and the dependence on$1/\epsilon$ from 4th to 2nd. One result not shown in the table is about a description of the solution. In \cite{gilyen2020improved}, the authors define an $s$-sparse description of a vector $\x$ as an $s$-sparse vector $\y$ such that $\x = A^\dag \y$. The sparsity of $\y$ determines the runtime of querying one entry of $\x$. In \cite{gilyen2020improved}, they obtain an $\widetilde{O}(\kappa_F^2\kappa^2/\epsilon^2)$-sparse description of the approximate solution in cost $\widetilde{O}(\kappa_F^6\kappa^2/\epsilon^4)$. In comparison, we obtain an $\widetilde{O}(\kappa_F^2)$-sparse description in cost $\widetilde{O}(\kappa_F^4\kappa^2/\epsilon^2)$ by Theorem \ref{thm:improvements}. 
    %
    \item When $A$ is row sparse, i.e., each row is sparse,  the randomized Kaczmarz method can find an approximate sparse solution in time $\widetilde{O}( s\kappa_F^2)$, where $s$ is the row sparsity. Moreover, the output is a vector, which is stronger than the outputs of quantum and quantum-inspired algorithms. Assuming $s=\widetilde{O}(1)$, this result means that the quantum speedup of solving row sparse linear systems is quadratic if using the QRAM data structure. It also means that to explore large quantum speedups in solving sparse linear systems, we may need to move our attentions to other quantum models.
    For example, it is known that when $A$ is row and column sparse, the quantum algorithm has complexity $\widetilde{O}(s\|A\|_{\max}\|A^{+}\|)$ in the sparse access input model \cite{chakraborty_et_al:LIPIcs:2019:10609}, which is better than the quantum result we show in the table. 
    %
    \item When $A$ is sparse and SPD, the randomized coordinate descent method can find an approximate sparse solution (also a sparse vector) in time $\widetilde{O}( s {\rm Tr}(A) \|A^{+}\|  )$. In this special case, it may happen that quantum computers do not achieve a speedup in solving linear systems. For example, suppose $A$ is a density matrix with $\|A\| = \Theta(1), s = \widetilde{O}(1)$, then the randomized coordinate descent method costs $\widetilde{O}(\kappa)$, which is also the cost of the quantum linear solver. As shown in \cite{orsucci2021solving}, to solve a sparse SPD linear system in a quantum computer, the dependence on $\kappa$ is at least linear. However, for some special SPD linear systems, quantum algorithms can achieve quadratic speedups in terms of $\kappa$.\footnote{In \cite{orsucci2021solving}, the authors mentioned two special cases. Case 1: Access to a block-encoding of $I-\eta A$ and $\|A^{+}\b\|\in O(\kappa)$, where $\eta \in(0,1]$. Case 2: $A$ is a sum of SPD local Hamiltonians, $\b$ is sparse and a parameter $\gamma \in O(1)$ that quantifies the overlap of $\b$ with the subspace where $A$ is non-singular.} However, the special cases mentioned in \cite{orsucci2021solving} seem not directly related to the case we mentioned above.
    \item In Table \ref{tab:comparison}, the complexity of the quantum algorithm is the cost to obtain the quantum state corresponding to the solution. For quantum-inspired algorithms, the table shows the complexity to obtain the data structure for the solution that allows sampling and querying. This output can be viewed as a classical analogue to the output of the quantum algorithm. If we want to estimate the norm of the solution, then there is an extra factor $1/\epsilon$ for the quantum algorithm \cite{chakraborty_et_al:LIPIcs:2019:10609} and an extra factor $1/\epsilon^2$ for the quantum-inspired algorithms \cite{chia2020sampling} (also see Lemma \ref{lem1}). However, for the randomized classical algorithms, the complexity does not change because their outputs are sparse vectors whose sparsity has the same order as the time complexity. So to estimate the norm of the solution, the quantum and quantum-inspired algorithms are worse than the randomized classical algorithms in terms of the precision in the row sparse case.
    \item
    For solving a row sparse linear system, the table shows that the complexities of quantum and classical randomized algorithms are logarithmic in the precision $\epsilon$. However, if we write down the complexity more precisely, the classical randomized algorithms may have better dependence on $\epsilon$, at least compared with the currently best quantum algorithm.
    Indeed, for the classical randomized algorithms, the dependence on $\epsilon$ is $O(\log(1/\epsilon))$, while it is $O(\log^2(1/\epsilon))$ for the quantum algorithm \cite{chakraborty_et_al:LIPIcs:2019:10609}. If the linear system is also column sparse,  the dependence on $\epsilon$ can be reduced to $O(\log(1/\epsilon))$ in the quantum case \cite{lin2020optimal}.
    %
\end{itemize}

\subsection{Summary of the techniques}

Our main algorithm is based on a generalization of Kaczmarz method. As briefly introduced above (also see Section \ref{sec:Randomized Kaczmarz algorithm} for more details), the Kaczmarz method aims to solve one linear constraint at each step of the iteration, which can be viewed as a particular case of the stochastic gradient method. To design an efficient quantum-inspired algorithm, we concentrate on the dual form of the Kaczmarz method. Namely, introducing $\y$ such that $\x = A^\dag \y$, then focusing on the iteration on $\y$. The linear system $A\x=\b$ can be reformulated as the least squares problem (LSP) $\x = \arg\min_{\x}\|A\x-\b\|$. It turns out that $\y$ converges to the optimal solution of the dual problem of the LSP.  Our idea is to apply the iterative method to find a sparse approximate solution $\y$ of the dual problem. The sparsity is guaranteed by the iterative method. Now $\x = A^\dag \y$ gives an approximate solution of the LSP. Since $\y$ is sparse, it is easy to query an entry of $\x$. To sample an entry of $\x$, we can use the rejection sampling idea. 

The technical part of this paper is the complexity analysis of the rejection sampling method. However, it turns out that the main cost of the above algorithm comes from the computation of $\y$, whose complexity analysis is straightforward. It is time-consuming if we directly use the iterative method because it relates to the calculation of inner products. To alleviate this, we modify the iteration by using the sampling idea. We prove that this modified method has the same convergence rate as the original one by choosing some parameters appropriately.

{\bf Notation.}
For a matrix $A$, we use $A^\dag$ to denote its conjugate transpose, $A^{+}$ to denote its Moore-Penrose inverse. 
The scaled condition number is $\kappa_F = \|A\|_F \|A^{+}\|$, and the condition number is $\kappa = \|A\|\|A^{+}\|$, where $\|\cdot\|_F$ is the Frobenius norm and $\|\cdot\|$ is the spectral norm. The row sparsity of $A$ is defined as the maximum of the number of nonzero zero entries of any row.
For a unit vector $\v$, we sometimes denote it as $\ket{\v}$, and denote $\v^\dag$ as $\bra{\v}$. For two vectors $\a,\b$, we use $\langle \a|\b\rangle$ or $\a\cdot \b$ to denote their inner product. If $m$ is an integer, we define $[m]:=\{1,2,\ldots,m\}$. For matrix $A$, we denote its $i$-th row as $A_{i*}$, its $j$-th column as $A_{*j}$ and the $(i,j)$-th entry as $A_{i,j}$. For a vector $\v$, we denote the $i$-th component as $v_i$. We use $\{\e_1,\ldots,\e_m\}$ to denote the standard basis of $\mathbb{C}^m$.

\section{Randomized Kaczmarz algorithm}
\label{sec:Randomized Kaczmarz algorithm}

Let $A$ be an $m\times n$ matrix, $\b$ be an $m\times 1$ vector.
Solving the linear system $A\x = \b$ can be described as a quadratic optimization problem, i.e., the so-called least squares problem (LSP) 
\be \label{LSP}
\arg\min_{\x \in \mathbb{C}^n } \quad \|A\x - \b\|^2.
\ee
Classically, there are many algorithms for solving the LSP \cite{golub13}. The best known classical algorithm might be the conjugate-gradient (CG) method. Compared to CG, the Kaczmarz algorithm has received much attention recently due to its simplicity. For instance, at each step of iteration, CG uses the whole information of the matrix $A$ and $O(mn)$ operations, while Kaczmarz method only uses one row of $A$ and  $O(n)$ operations. So Kaczmarz method becomes effective when the entire matrix (data set) cannot be loaded into memory.

\begin{defn}[Randomized Kaczmarz algorithm \cite{strohmer2009randomized}]

Let $\x_0$ be an arbitrary initial approximation to the solution of (\ref{LSP}). For $k = 0, 1, \ldots$, compute
\be \label{Randomized Kaczmarz algorithm}
\x_{k+1}  =  \x_k  + \frac{b_{r_k} - \langle A_{r_k*}| \x_k \rangle   }{\|A_{r_k*}\|^2} A_{r_k*},
\ee
where $A_{r_k*}$ is the $r_k$-th row of $A$, and $r_k$ is chosen from $[m]$ randomly with probability proportional to $\|A_{r_k*}\|^2$.

\end{defn}

The Kaczmarz algorithm has a clear geometric meaning, see Figure \ref{fig:Illustration of Kaczmarz algorithm}. Geometrically, $\x_{k+1}$ is the orthogonal projection of $\x_k$ onto the hyperplane $A_{r_k*}\cdot \x = b_{r_k}$. So at each step of iteration, this method finds the closest vector that satisfies the chosen linear constraint.

\begin{figure}[h]
\centering

\begin{tikzpicture}[scale=0.8]
\draw[line width=0.3mm] (-3/2,0) -- (7,0) 
node[label=right:{$A_{1*} \cdot \x = b_1$}]{}; 
\draw[line width=0.3mm] (-1,-1) -- (4,4)
node[label=right:{$A_{2*} \cdot \x = b_2$}]{}; 
\draw[line width=0.3mm,dashed] 
(6,3) node[circle,fill=black,minimum size=0.1pt,scale=0.5,label=right:{$\x_0$}]{} -- 
(6,0) node[circle,fill=black,minimum size=0.1pt,scale=0.5,label=below:{$\x_1$}]{}; 
\draw[line width=0.3mm,dashed] 
(6,0) -- 
(3,3) node[circle,fill=black,minimum size=0.1pt,scale=0.5,label=left:{$\x_2$}]{}; 
\draw[line width=0.3mm,dashed] 
(3,3) -- 
(3,0) node[circle,fill=black,minimum size=0.1pt,scale=0.5,label=below:{$\x_3$}]{}; 
\draw[line width=0.3mm,dashed] 
(3,0) -- 
(3/2,3/2) node[circle,fill=black,minimum size=0.1pt,scale=0.5,label=left:{$\x_4$}]{}; 
\draw[line width=0.3mm,dashed] 
(3/2,3/2) -- 
(3/2,0) node[circle,fill=black,minimum size=0.1pt,scale=0.5,label=below:{$\x_5$}]{}; 
\draw[line width=0.3mm,dashed] 
(3/2,0) -- 
(3/4,3/4) node[circle,fill=black,minimum size=0.1pt,scale=0.5,label=left:{$\x_6$}]{}; 
\draw[line width=0.3mm,dashed] 
(3/4,3/4) -- 
(3/4,0); 
\draw
(0,0) node[circle,fill=black,minimum size=0.1pt,scale=0.5,label=below:{$\x_*$}]{}; 
\end{tikzpicture}

\caption{Illustration of Kaczmarz algorithm when $m=2$, where $\x_*$ is the optimal solution.}
\label{fig:Illustration of Kaczmarz algorithm}
\end{figure}

From the definition, we can simplify the notation by defining
\be
\tilde{\b} = {\rm diag}(\|A_{i*}\|^{-1}:i\in[m]) \b, \quad
\widetilde{A} = {\rm diag}(\|A_{i*}\|^{-1}:i\in[m]) A,
\ee
so that
$\tilde{b}_i = b_{i}/\|A_{i*}\| $ and 
$\widetilde{A}_{i*} = A_{i*}/\|A_{i*}\|$.
Then we can rewrite 
(\ref{Randomized Kaczmarz algorithm}) as
\be \label{projection}
\x_{k+1}  =  \x_k  + \left( \tilde{b}_{r_k} - \langle \widetilde{A}_{r_k*} | \x_k \rangle  \right) \widetilde{A}_{r_k*}
= \left(I - \ket{\widetilde{A}_{r_k*}} \bra{\widetilde{A}_{r_k*}} \right) \x_k + \tilde{b}_{r_k} \widetilde{A}_{r_k*} .
\ee
In \cite{shao2020row}, the Kaczmarz algorithm has been quantized based on the description (\ref{projection}). As an orthogonal projector, it is not hard to extend $I - \ket{\widetilde{A}_{i*}} \bra{\widetilde{A}_{i*}}$ to a block of a unitary 
operator. For example, a unitary of the form $I\otimes (I - \ket{\widetilde{A}_{i*}} \bra{\widetilde{A}_{i*}}) + X \otimes \ket{\widetilde{A}_{i*}} \bra{\widetilde{A}_{i*}}$, where $X$ is Pauli-X.  
This unitary is closely related to the Grover diffusion operator when we apply it to $|-\rangle|\x\rangle$.

\begin{lem}[Theorem 2 of  \cite{strohmer2009randomized}]
\label{lem:converge rate}
Assume that $A\x=\b$ is consistent.
Let $\x_*=A^+\b$ be a solution of  (\ref{LSP}). Then the randomized Kaczmarz algorithm converges to $\x_*$ in expectation, with the average error
\be
\E[\|\x_T - \x_*\|^2] 
\le (1 - \kappa_F^{-2})^T \|\x_0 - \x_*\|^2,
\ee
where $\kappa_F = \|A\|_F \|A^{+}\|$.
\end{lem}

The above result is stated for consistent linear systems, while the Kaczmarz algorithm also works for inconsistent linear systems, for example, see \cite{moorman2020randomized}. The convergence rate does not change, the error bound is affected by an extra term depending on $\min\|A\x-\b\|$. Our results (especially the main Theorem \ref{thm:improvements}) stated below are indeed for general cases -- consistent and inconsistent.

The above result holds for all $\x_0$,  so throughout this paper, we shall assume $\x_0=0$ for simplicity.
In the above lemma, by Markov's inequality, with high probability 0.99, we have $\|\x_T - \x_*\|^2 \le 100(1 - \kappa_F^{-2})^T \| \x_*\|^2$. 
To make sure the error is bounded by $\epsilon^2\| \x_*\|^2$ with high probability,
it suffices to choose 
\be
T = O(\kappa_F^2 \log(1/\epsilon)).
\ee

From (\ref{Randomized Kaczmarz algorithm}), we know that there exist $y_{k,0},\ldots,y_{k,k-1}$ such that
$
\x_k = \sum_{j=0}^{k-1} y_{k,j} A_{r_j*},
$
that is $\x_k = A^\dag \y_k$ for some vector $\y_k$.  One obvious fact is that $\y_k$ is at most $k$-sparse. In \cite{gilyen2020improved}, $\y_k$ is called the sparse description of $\x_k$. So this description comes easily in the Kaczmarz method. Actually, $\y_k$ converges to the optimal solution of the dual problem of the LSP (\ref{LSP}).
More precisely, an alternative formulation of LSP is to find the least-norm solution of the linear system
\[
\min_{\x\in \mathbb{C}^n} \quad \frac{1}{2}\|\x\|^2, \quad {\rm s.t.} \quad  A\x=\b.
\]
The dual problem takes the form:
\be \label{dual LSP}
\min_{\y\in \mathbb{C}^m}\quad g(\y) := \frac{1}{2}\|A^\dag \y\|^2 - \b\cdot \y.
\ee
The randomized Kaczmarz algorithm (\ref{Randomized Kaczmarz algorithm}) is equivalent to one step of the stochastic gradient descent  method \cite{nemirovski2009robust} applied to $\frac{1}{2}(A_{r_k*}\cdot \x - b_{r_k})^2$ with stepsize $1/\|A_{r_k*}\|^2$. It is also equivalent to one step of the randomized coordinate descent method \cite{nesterov2012efficiency} applied to the dual problem (\ref{dual LSP}). Namely, when we apply this method to the $r_k$-th component, the gradient is $\nabla_{r_k} g = \langle A_{r_k*} |  A^\dag| \y_{k} \rangle - b_{r_k}$. Setting the stepsize as $1/\|A_{r_k*}\|^2$ leads to the updating 
\be \label{updating for sparse case}
\y_{k+1} = \y_{k} +  \frac{b_{r_k}  - \langle A_{r_k*} |  A^\dag| \y_{k} \rangle }{\|A_{r_k*}\|^2}  \e_{r_k},
\ee
where $\{\e_1,\ldots,\e_m\}$ is the standard basis of $\mathbb{C}^m$.
We can recover the original Kaczmarz iteration (\ref{Randomized Kaczmarz algorithm}) by multiplying $A^\dag$ on both sides of (\ref{updating for sparse case}).


The Kaczmarz method can be generalized to use multiple rows at each step of iteration \cite{gower2015randomized}. Instead of orthogonally projecting to a hyperplane determined by one linear constraint, we can orthogonally project to a vector space defined by several linear constraints. This idea is similar to the minibatch stochastic gradient descent. However, this generalization needs to compute the pseudoinverse of certain matrices. A simple pseudoinverse-free variant is as follows \cite{moorman2020randomized}:
\be \label{Kaczmarz method in average}
\x_{k+1} = \x_k + \frac{\alpha}{2q}\sum_{i \in \T_k} \frac{b_{i} - \langle A_{i*}| \x_k \rangle   }{\|A_{i*}\|^2} A_{i*},
\ee
where $\T_k$ is a random set of $q$ row indices sampled with replacement, $\alpha$ is known as the relaxation parameter. 
Each index $i\in[m]$ is put into $\T_k$ with probability proportional to $\|A_{i*}\|^2$. For instance, let $\alpha = 2$, then (\ref{Kaczmarz method in average}) can be rewritten in a parallel form
\bes
\x_{k+1} = \frac{1}{q} \sum_{i \in \T_k} \left( \x_k +   \frac{b_{i} - \langle A_{i*}| \x_k \rangle   }{\|A_{i*}\|^2} A_{i*} \right) .
\ees
It was proved in \cite[Corollary 3]{moorman2020randomized} that if $\alpha = q = \|A\|_F^2/\|A\|^2$ (also see a related proof in Appendix \ref{app for the convergence rate}), 
\be
\E[\|\x_T - \x_*\|^2] \leq
\left( 1 - \kappa^{-2}  \right)^T  \|\x_0 - \x_*\|^2.
\ee
So the number of iterations required to achieve error $\epsilon$ via (\ref{Kaczmarz method in average}) is $O(\kappa^2\log(1/\epsilon))$.

\section{Preliminaries of quantum-inspired algorithms}
\label{sec:Preliminaries of quantum-inspired algorithms}

This section briefly recalls some definitions and results about quantum-inspired algorithms that will be used in this paper. The main reference is \cite{chia2020sampling}.

\begin{defn}[Query access] 
For a vector $\v \in \mathbb{C}^n$, we have $Q(\v)$, query access to $\v$ if for all $i\in[n]$, we can query $v_i$. Likewise, for a matrix $A\in \mathbb{C}^{m\times n}$, we have $Q(A)$ if for all $(i, j) \in [m] \times [n]$, we can query $A_{ij}$. Let $\q(\v)$ (or $\q(A)$) denote the (time) cost of such a query.
\end{defn}

\begin{defn}[Sampling and query access to a vector]
For a vector $\v \in \mathbb{C}^n$, we have $SQ(\v)$, sampling and query access to $\v$, if we can:
\begin{enumerate}
    \item Sample: obtain independent samples $i\in[n]$ following the distribution $\mathcal{D}_{\v} \in \mathbb{R}^n$, where $\mathcal{D}_{\v}(i) = |v_i|^2/\|\v\|^2$.
    \item Query: query entries of $\v$ as in $Q(\v)$;
    \item Norm: query $\|\v\|$.
\end{enumerate}
Let $\q(\v), \s(\v)$, and $\n(\v)$ denote the cost of querying entries, sampling indices, and querying the norm respectively. Further define $\s\q(\v) := \q(\v) + \s(\v) + \n(\v)$.
\end{defn}

\begin{defn}[Sampling and query access to a matrix]
\label{defn:Data structure for matrices}

For any matrix $A \in \mathbb{C}^{m\times n}$, we have $SQ(A)$ if we have $SQ(A_{i*})$ for all $i\in[m]$ and $SQ(\a)$, where $\a$ is the vector of row norms. We define $\s\q(A)$ as the cost to obtain $SQ(A)$.
\end{defn}

The sampling and query access defined above can be viewed as a classical analogue to the QRAM data structure, which is widely used in many quantum machine learning (QML) algorithms. In QML, this data structure can be used to efficiently prepare the quantum states of vectors, which are usually the inputs of the QML algorithms. For matrices, it allows us to find an efficient block-encoding so that many quantum linear algebraic techniques become effective. However, the QRAM (and Definition \ref{defn:Data structure for matrices}) is a quite strong data structure because it contains a lot of information (e.g. $\|A_{i*}\|$ and $\|A\|_F$) about $A$, which is usually not easy to obtain. For instance, it usually needs $\Theta(\|A\|_0)$ operations to build this data structure, which $\|A\|_0$ is the number of nonzero entries of $A$. For quantum-inspired algorithms, we need to use certain sampling techniques to make sure that the complexity has a polylog dependence on the dimension, so it seems quite natural to use the above data structure. In comparison, there are also some other efficient quantum models, e.g., sparse-access model, in which the QML algorithms are still efficient.

The sampling operation on $\v$ can be viewed as a classical analogue of the measurement of the quantum state $\ket{\v}=\|\v\|^{-1} \sum_{i\in[n]} v_i \ket{i}$ in the computational basis.  For some problems, we do not have the state $\ket{\v}$ exactly. Instead, what we have is a state of the form $\sin(\theta) \ket{\v} \ket{0} + \cos(\theta) \ket{\w}\ket{1}$. Then the probability of obtaining $\ket{i,0}$ is $\frac{v_i^2}{\|\v\|^2}\sin^2(\theta)$.
This leads to the definition of oversampling for quantum-inspired algorithms.

\begin{defn}[Oversampling and query access]
\label{defn:Oversampling and query access}
For $\v \in \mathbb{C}^n$ and $\phi\geq 1$, we have $SQ_{\phi}(\v)$, $\phi$-oversampling and query access to $\v$, if we have $Q(\v)$ and $SQ(\tilde{\v})$ for  a vector $\tilde{\v} \in \mathbb{C}^n$ satisfying $\|\tilde{\v}\|^2 = \phi \|\v\|^2$ and $|\tilde{v}_i|^2 \geq |v_i|^2$ for all $i\in[n]$.
We denote $\s_{\phi}(v) = \s(\tilde{\v}), \q_{\phi}(v) = \q(\tilde{\v}), \n_{\phi}(v) = \n(\tilde{\v})$, and $\s\q_{\phi}(\v) = 
\s_{\phi}(v)+\q_{\phi}(v)+\q(v)+\n_{\phi}(v)$.

\end{defn}

For example, if we view $\v$ as the vector $(\v,0,\ldots,0)$ and $\tilde{\v}$ as the vector $(\v, \|\v\|\cot(\theta) \w)$ so that they have the same dimension, then it is easy to see that $\|\tilde{\v}\|^2 = \|\v\|^2/\sin^2(\theta)$ and $\phi = 1/\sin^2(\theta)$. The condition $|\tilde{v}_i|^2 \geq |v_i|^2$ is also satisfied. But the above definition is stronger than this intuition.

\begin{lem}[Lemma 2.9 of \cite{chia2020sampling}]
\label{lem1}

Given $SQ_{\phi}(\v)$ and $\delta \in(0,1]$ we can sample from $\mathcal{D}_{\v}$ with success probability at least $1-\delta$ and cost $O(\phi \s\q_{\phi}(\v) \log(1/\delta))$. We can also estimate $\|\v\|$ up to relative error
$\epsilon$ in time $O(\epsilon^{-2}\phi \s\q_{\phi}(\v) \log(1/\delta))$.
\end{lem}

\begin{lem}[Lemma 2.10 of \cite{chia2020sampling}]
\label{lem2}

Given $SQ_{\varphi_i}(\v_i)$ and $\lambda_i \in \mathbb{C}$ for $i\in[k]$, denote $\v = \sum_i \lambda_i \v_i$, then we have
$SQ_{\phi}(\v)$ for
$\phi = k \sum_i \varphi_i \|\lambda_i \v_i\|^2/ \|\v\|^2$ and $\s\q_{\phi}(\v) = \max_i \s_{\varphi_i}(\v_i) + \sum_i \q(\v_i)$.

\end{lem}

The first lemma tells us how to construct the original sampling access from an over-sampled one.
By Definition \ref{defn:Oversampling and query access}, given $SQ_{\phi}(\v)$, we have $Q(\v)$. Lemma \ref{lem1} tells us how to obtain $S(\v)$ from $SQ_{\phi}(\v)$. So this lemma actually tells us how to obtain $SQ(\v)$ from $SQ_{\phi}(\v)$. 
The second lemma is about how to construct a linear combination of the data structures. Equivalently, it is about constructing the data structure for $M\c$ when we have the data structures of $M$ and $\c$, where $M$ is a matrix and $\c$ is a vector.

In our algorithms below, $\v_i$ will be a row of $A$, $k$ will refer to the sparsity of a certain known vector and $\v$ will be the approximate solution to the linear system. So in Lemma \ref{lem2}, $\s\q_{\phi}(\v) = O(k\s\q(A))$. Combing the above two lemmas, given $SQ(\v_i)$ and $\lambda_i$, the cost to sample and query $\v$ is $\widetilde{O}(k \phi\s\q(A))$, and the cost to estimate its norm is $\widetilde{O}(k \phi\s\q(A)/\epsilon^2)$. So a main difficulty in the the complexity analysis  is the estimation of $\phi$.


\section{The algorithm for solving linear systems}
\label{section:The algorithm for solving linear systems}

\subsection{Solving general linear systems}

In this section, based on the Kaczmarz method, we give a quantum-inspired classical algorithm for solving linear systems. The basic idea is similar to \cite{gilyen2020improved}. However, our result has a lower complexity and the analysis is much simpler.
The main idea is as follows: we first use the iteration (\ref{updating for sparse case}) to compute $\y_T$, then output $SQ(\x_T)$ with $\x_T=A^\dag \y_T$ by Lemmas \ref{lem1} and \ref{lem2}.

When $A$ is dense, the calculation of $\langle A_{r_k*}  | A^\dag| \y_{k} \rangle$ in (\ref{updating for sparse case}) is expensive. 
Note that $\langle A_{r_k*}  | A^\dag| \y_{k} \rangle = \sum_{j=1}^n A_{r_k,j} \langle A_{*j}|\y_k \rangle$. As an inner product, we can use the sampling idea (e.g. Monte Carlo) to approximate it.
Now let $j_1, \ldots, j_d$ be drawn randomly from $[n]$ under the distribution that ${\rm Pr}[j] = \|A_{*j}\|^2/\|A\|_F^2$. This distribution is designed to put more weight on the most important rows \cite{drineas2006fast,strohmer2009randomized}. 
Let $S = \{j_1, \ldots, j_d\}$ be the sequence of the chosen indices.
We modify the updating of $\y$ to:
\be \label{updating for y'}
\y_{k+1}  = \y_{k}  +  \frac{1 }{\|A_{r_k*}\|} \left( \tilde{b}_{r_k} - \frac{1}{d}\sum_{j\in S} \tilde{A}_{r_k,j} \langle A_{*j}|\y_k  \rangle \frac{\|A\|_F^2}{\|A_{*j}\|^2}
\right) \e_{r_k}.
\ee

For conciseness, define the matrix $D$ as
\be \label{sampling matrix}
D = 
\sum_{j\in S}
\frac{\|A\|_F^2}{d \|A_{*j}\|^2} \ket{j} \bra{j}.
\ee 
Then
\bea
\y_{k+1} &=& \y_k  + \frac{\tilde{b}_{r_k} - \langle \tilde{A}_{r_k*} | D A^\dag| \y_{k} \rangle }{\|A_{r_k*}\|}  \e_{r_k} \\
&=& \y_k  + \frac{\tilde{b}_{r_k} - \langle \tilde{A}_{r_k*} | A^\dag| \y_{k} \rangle }{\|A_{r_k*}\|}  \e_{r_k} + \frac{\langle \tilde{A}_{r_k*} |(I-D) A^\dag| \y_{k} \rangle }{\|A_{r_k*}\|}  \e_{r_k}.  
\label{updating for y'-matrix form}
\eea
Compared to the original iteration (\ref{updating for sparse case}), the procedure (\ref{updating for y'-matrix form}) contains a perturbation term.
To analyze its convergence rate, we need to determine $d$.
Before that, we first analyze the properties of the following random variable:
\be \label{notation of mu}
\mu_k := \langle \tilde{A}_{r_k*} |(I-D) A^\dag| \y_{k} \rangle = \langle\tilde{A}_{r_k*} | I - D |\x_{k} \rangle =
\langle\tilde{A}_{r_k*} | \x_{k} \rangle
-
\frac{1}{d}\sum_{j\in S} \tilde{A}_{r_k,j} x_{k,j} \frac{\|A\|_F^2}{\|A_{*j}\|^2}.
\ee

\begin{lem} 
\label{lem:mean and variance}
The mean and the variance of $\mu_k$ satisfy
$
\E_D[\mu_k] = 0,
{\rm Var}_D[\mu_k] \leq  \frac{1}{d}  \sum_{j=1}^n 
\tilde{A}_{r_k,j}^2 x_{k,j}^2 \frac{\|A\|_F^2}{\|A_{*j}\|^2}  .
$
\end{lem}

\begin{proof} The results follow from the standard calculation.
\beas
\E_D[\mu_k] &=& \langle\tilde{A}_{r_k*} | \x_{k} \rangle - \E_D\left[  \frac{1}{d} \sum_{j\in S} \tilde{A}_{r_k,j} x_{k,j} \frac{\|A\|_F^2}{\|A_{*j}\|^2}\right]  \\
&=&  \langle\tilde{A}_{r_k*} | \x_{k} \rangle - \E_D\left[  \tilde{A}_{r_k,j} x_{k,j} \frac{\|A\|_F^2}{\|A_{*j}\|^2} \right]  \\
&=&  \langle\tilde{A}_{r_k*} | \x_{k} \rangle - \sum_{j=1}^n \tilde{A}_{r_k,j} x_{k,j} \frac{\|A\|_F^2}{\|A_{*j}\|^2} \frac{\|A_{*j}\|^2}{\|A\|_F^2} \\
&=& 0. \\
{\rm Var}_D[\mu_k] &=&  {\rm Var}_D\left[\frac{1}{d}\sum_{j\in S} \tilde{A}_{r_k,j}  x_{k,j} \frac{\|A\|_F^2}{\|A_{*j}\|^2}\right] \\
&=& \sum_{j\in S} {\rm Var}_D\left[\frac{1}{d} \tilde{A}_{r_k,j}  x_{k,j} \frac{\|A\|_F^2}{\|A_{*j}\|^2}\right] \\
&\leq& \frac{1}{d} \E_D\left[\left( \tilde{A}_{r_k,j}  x_{k,j} \frac{\|A\|_F^2}{\|A_{*j}\|^2}\right)^2\right] \\
&=& \frac{1}{d}  \sum_{j=1}^n 
\left(\tilde{A}_{r_k,j}  x_{k,j} \frac{\|A\|_F^2}{\|A_{*j}\|^2} \right)^2
\frac{\|A_{*j}\|^2}{\|A\|_F^2}\\
&=& \frac{1}{d}  \sum_{j=1}^n 
\tilde{A}_{r_k,j}^2 x_{k,j}^2 \frac{\|A\|_F^2}{\|A_{*j}\|^2} .
\eeas
This completes the proof.
\end{proof}

\begin{rem}{\rm
In an early version, we further bounded ${\rm Var}_D[\mu_k]$ by $\frac{1}{d}  \frac{\|A\|_F^2}{\min_{j \in [n]} \|A_{*j}\|^2} \|\x_k\|^2$. This leads to $d=\widetilde{O}(\kappa_F^4/\epsilon^2)$. In comparison, the current bound leads to $d=\widetilde{O}(\kappa_F^2/\epsilon^2)$. So, the overall complexity of the algorithms presented below can be reduced by a factor of $\kappa_F^2$.
}\end{rem}

From (\ref{updating for y'-matrix form}), the updating of $\x_k  = A^\dag \y_k $ then becomes
\be
\x_{k+1}  = \x_k + (\tilde{b}_{r_k} - \langle \tilde{A}_{r_k*} |\x_{k} \rangle )  \tilde{A}_{r_k*}
+ \langle\tilde{A}_{r_k*} | I - D |\x_{k} \rangle \tilde{A}_{r_k*} . 
\label{randomized procedure}
\ee
Next, we prove a similar result to Lemma \ref{lem:converge rate}. The idea of our proof is similar to that of Lemma \ref{lem:converge rate} from \cite{strohmer2009randomized}, which nicely utilizes the geometrical structure of the Kaczmarz method.

\begin{prop} 
\label{prop:convergence}
Choose $d = 4 (\kappa_F^2/\epsilon^{2}) (\log 2/\epsilon^2)$, then after $T = \kappa_F^2 (\log 2/\epsilon^2)$ steps of iteration of the procedure (\ref{randomized procedure}), we have $\E[\|\x_{T}  - \x_*\|^2] \leq \epsilon^2 \|\x_*\|^2 $.
\end{prop}

\begin{proof}
For any vector $\v$ such that $A\v\neq 0$ we have $\|\v\|^2 \leq \|A^{+}\|^2 \|A\v\|^2$. This implies
\be \label{proof:eq1}
\frac{\|\v\|^2}{\kappa_F^2} =
 \frac{\|\v\|^2}{\|A\|_F^2 \|A^{+}\|^2} \leq \sum_{i=1}^m \langle \tilde{A}_{i*}| \v\rangle^2 \frac{\|A_{i*}\|^2}{\|A\|_F^2}
 =\E[\langle \tilde{A}_{i*}| \v\rangle^2].
\ee
Here the random variable is defined by
\[
{\rm Pr}[i] = \frac{\|A_{i*}\|^2}{\|A\|_F^2}.
\]

For simplicity, we set the initial vector $\x_0  = 0$.
In procedure (\ref{randomized procedure}), we denote
$\tilde{\x}_{k+1} = \x_k + (\tilde{b}_{r_k} - \langle \tilde{A}_{r_k*} |\x_{k} \rangle )  \tilde{A}_{r_k*}$, it is the orthogonal projection of $\x_k$ onto the hyperplane $\tilde{A}_{r_k*} \cdot \x = \tilde{b}_{r_k} $. Thus we can rewrite the iteration as
\[
\x_{k+1} = \tilde{\x}_{k+1} + \mu_k \tilde{A}_{r_k*}.
\]
The orthogonality implies (see Figure \ref{fig:Illustration of the proof})
\beas
\|\x_{k+1} - \x_*\|^2 
&=&  \|\tilde{\x}_{k+1} - \x_*\|^2 + \mu_k^2 \\
&=&  \|\x_{k} - \x_*\|^2 - \|\x_k - \tilde{\x}_{k+1}\|^2 + \mu_k^2 \\
&=&  \|\x_{k} - \x_*\|^2 - 
\langle \x_{k} - \x_*| \tilde{A}_{r_k*} \rangle^2
+ \mu_k^2 \\
&=&  \left(1 - 
\langle \frac{\x_{k} - \x_*}{\|\x_{k} - \x_*\|}| \tilde{A}_{r_k*} \rangle^2 \right) \|\x_{k} - \x_*\|^2 + \mu_k^2.
\eeas

\begin{figure}
\centering

\begin{tikzpicture}[>=stealth]
\draw[line width=0.3mm] (-2,-1) -- (-1,0)  node[circle,fill=black,minimum size=0.1pt,scale=0.5,label=right:{$\x_*$}]{} -- (2.3,3.3) node[label=right:{$\tilde{A}_{r_k*} \cdot \x = \tilde{b}_{r_k}$}]{}; 
\draw[line width=0.3mm,dashed] (0,3) node[circle,fill=black,minimum size=0.1pt,scale=0.5,label=above:{$\x_k$}]{} -- 
(1,2) node[circle,fill=black,minimum size=0.1pt,scale=0.5,label=left:{$\tilde{\x}_{k+1}$}]{}; 
\draw[line width=0.3mm,dashed] (1,2)  -- 
(4/3,5/3) node[circle,fill=black,minimum size=0.1pt,scale=0.5,label=right:{$\x_{k+1}$}]{}; 
\draw[line width=0.3mm,dashed] (-1,0)  --  (0,3); 
\draw[line width=0.3mm,dashed] (-1,0)  --  (4/3,5/3); 
\draw[->] (7/6+1/20,11/6+1/20) arc (130:70:60pt) node[minimum size=0.1pt,scale=0.5,label=right:{$\mu_k$}]{};
\end{tikzpicture}

\caption{Illustration of the orthogonality in the proof of Proposition \ref{prop:convergence}.}
\label{fig:Illustration of the proof}
\end{figure}

Taking the expectation yields
\bes
\E[\|\x_{k+1} - \x_*\|^2]
=
\left(1 - \E\left[
\langle \frac{\x_{k} - \x_*}{\|\x_{k} - \x_*\|}| \tilde{A}_{r_k*} \rangle^2 \right]\right) \E[\|\x_{k} - \x_*\|^2] + \E[\mu_k^2].
\ees
In the above, the expectation is taken first on $D$ by fixing $k, r_k$, then on $r_k$ by fixing $k$, and finally on $k$. Also if $A(\x_{k} - \x_*) = 0$, then $\x_k$ is already the optimal solution, so we assume that this does not happen before convergence.

By equation (\ref{proof:eq1}),
\[
1 - \E\left[
\langle \frac{\x_{k} - \x_*}{\|\x_{k} - \x_*\|}| \tilde{A}_{r_k*} \rangle^2 \right] \leq 1-\kappa_F^{-2}.
\]
By Lemma \ref{lem:mean and variance}, 
\beas
\E_{r_k} \Big[\E_D[\mu_k^2]\Big]
&\leq& \E_{r_k} \Bigg[\frac{1}{d}  \sum_{j=1}^n 
\tilde{A}_{r_k,j}^2 x_{k,j}^2 \frac{\|A\|_{\F}^2}{\|A_{*j}\|^2} \Bigg] \\
&=& \frac{1}{d} \sum_{r_k=1}^m \sum_{j=1}^n 
\tilde{A}_{r_k,j}^2 x_{k,j}^2 \frac{\|A\|_{\F}^2}{\|A_{*j}\|^2} \frac{\|A_{r_k*}\|^2}{\|A\|_{\F}^2} \\
&=& \frac{\|\x_k\|^2}{d}.
\eeas
Hence,
\beas
\E[\|\x_{k+1} - \x_*\|^2]
&\leq&
\left(1 - \kappa_F^{-2} \right) \E[\|\x_{k} - \x_*\|^2] + \frac{1}{d} \E[\|\x_k\|^2] \\
&\leq&
\left(1 - \kappa_F^{-2} \right) \E[\|\x_{k} - \x_*\|^2] + \frac{2}{d} (\|\x_*\|^2 + \E[\|\x_k-\x_*\|^2]) \\
&=&
(1 - \kappa_F^{-2} + \frac{2}{d} ) \E[\|\x_{k}  - \x_*\|^2] + \frac{2}{d} \|\x_*\|^2 .
\eeas
Now we choose $d =  \frac{4 \kappa_F^2 (\log 2/\epsilon^2)}{\epsilon^{2} }  $, then the above estimation leads to
\bes
\E[\|\x_{k+1} - \x_*\|^2]
\leq
\left(1 - \kappa_F^{-2} + \frac{1}{2}\kappa_F^{-2} \epsilon^2 (\log 2/\epsilon^2)^{-1} \right) \E[\|\x_{k} - \x_*\|^2] + \frac{\epsilon^2}{2\kappa_F^2 (\log 2/\epsilon^2)} \|\x_*\|^2 .
\ees
Setting $\rho := 1 - \kappa_F^{-2} + \frac{1}{2}\kappa_F^{-2}  \epsilon^2 (\log 2/\epsilon^2)^{-1}$ which is less than 1, we finally obtain
\beas
\E[\|\x_{T} - \x_*\|^2]
&\leq&
\rho^T \E[\|\x_{0} - \x_*\|^2] + \frac{\epsilon^2}{2\kappa_F^2 (\log 2/\epsilon^2)} \|\x_*\|^2 \sum_{i=0}^{T-1} \rho^i \\
&\leq&
\rho^T \|\x_*\|^2 + \frac{T\epsilon^2}{2\kappa_F^2 (\log 2/\epsilon^2)} \|\x_*\|^2.
\eeas
Let $T = \kappa_F^2 (\log 2/\epsilon^2)$, then
\beas
\rho^T &=&
\exp( \kappa_F^2 (\log 2/\epsilon^2) \log 
(1 - \kappa_F^{-2} + \kappa_F^{-2}  \epsilon^2 (\log 1/\epsilon)^{-1}) ) \\
&\approx&
\exp( - (\log 2/\epsilon^2) 
(1-\epsilon^2 (\log 1/\epsilon)^{-1})) ) \\
&\approx& \epsilon^2/2.
\eeas
Therefore, $\E[\|\x_{T} - \x_*\|^2] \leq \epsilon^2 \|\x_*\|^2$ as claimed.
\end{proof}

\begin{thm}
\label{thm:non sparse case}
Assume that $\x_* = A^+\b$.
Given $SQ(A), Q(\b)$, there is an algorithm that returns $SQ(\x)$ such that $\|\x - \x_*\| \leq \epsilon \|\x_*\|$ with probability at least 0.99 and 
\be
\s\q(\x) 
=\widetilde{O}( (\frac{\kappa_F^6}{\epsilon^2} + \alpha \kappa_F^4\kappa^4) (\s\q(A) + \q(\b))),
\ee
where $\alpha = \eta/(1-\eta^2)$ and $\eta=\|A\x_*-\b\|/\|\b\|$.
\end{thm}

\begin{proof}
Let $d = \widetilde{O}(\kappa_F^2/\epsilon^2) $. 
By the updating formula (\ref{updating for y'}), we know that $\y_k $ is at most $k$-sparse. So in the $k$-th step of the iteration, the calculation of the summation in (\ref{updating for y'}) costs $O(kd)$. At each step of the iteration, we sample a row of $A$ and query the corresponding entry of $\b$. So the overall cost of $k$-th step is $O(kd(\s\q(A) + \q(\b)) )$. After $T=O(\kappa_F^2\log(1/\epsilon))$ steps, it converges and the cost is $O(T^2d(\s\q(A) + \q(\b)) )=\widetilde{O}((\kappa_F^6/\varepsilon^2 )(\s\q(A) + \q(\b)) )$.

Note that $\x_T =A^\dag \y_T$ and $\y_T$ is $T$-sparse. By Lemma \ref{lem2} (taking $\v_i$ to be rows of $A$), we can obtain $SQ_{\phi}(\x_T)$ in time $O(T\s\q(A))$, where
\bes
\phi = T \frac{\sum_{i} y_{T,i}^2 \|A_{i*}\|^2}{\|\x_T\|^2}.
\ees
In Appendix \ref{app for non-sparse matrices}, we shall show that $\phi = O(\kappa_F^4 + \alpha \kappa_F^2 \kappa^4)$.
By Lemma \ref{lem1}, we can obtain $SQ(\x_T)$ in time
$
O(\phi T\s\q(A)) = O((\kappa_F^6 + \alpha \kappa_F^4 \kappa^4)\s\q(A)).
$
\end{proof}

In the above, the parameter $\eta$ describes the overlap of $\b$ in the column space of $A$. If $\b$ almost lies in the column space of $A$, then $\eta$ is close to 0. In this case, the complexity is basically dominated by the first term.

We next use the Kaczmarz method with averaging (\ref{Kaczmarz method in average}) to reduce the dependence on $\kappa_F$ in the above theorem. 
Similar to (\ref{randomized procedure}), we change the updating rule into
\be \label{kaczmarz average:dense}
\x_{k+1} = \x_k + \frac{1}{2}\sum_{i\in\T_k} (\tilde{b}_{i} - \langle \tilde{A}_{i*} |\x_{k} \rangle )  \tilde{A}_{i*}
+ \frac{1}{2} \sum_{i\in\T_k} \langle\tilde{A}_{i*} | I - D_i |\x_{k} \rangle \tilde{A}_{i*}, 
\ee
The definition of $D_i$ is similar to (\ref{sampling matrix}), which depends on $i$ now.
In Appendix \ref{app for the convergence rate}, we shall prove that if
$d = \widetilde{O}(\kappa_F^2/\epsilon^2)$,
then after $T = O(\kappa^2 \log(1/\epsilon))$ steps of iteration, we have $\E[\|\x_T-\x_*\|^2]\leq \epsilon^2 \|\x_0 - \x_*\|^2$. Together with the estimation of $\phi$ in Appendix \ref{app:Estimation of phi for Kaczmarz with averaging}, we have the following improved result.

\begin{thm}
\label{thm:improvements}
Assume that $\x_* = A^+\b$.
Given $SQ(A), Q(\b)$, there is an algorithm that returns $SQ(\x)$ such that $\|\x - \x_*\| \leq \epsilon \|\x_*\|$ with probability at least 0.99 and
\be
\s\q(\x) = \widetilde{O}\left( (\frac{\kappa_F^4\kappa^2}{\epsilon^2} + \alpha \kappa_F^2 \kappa^6 ) (\s\q(A) + \q(\b))  \right),
\ee
where $\alpha = \eta/(1-\eta^2)$ and $\eta=\|A\x_*-\b\|/\|\b\|$.
\end{thm}

\begin{proof}
From the updating formula (\ref{kaczmarz average:dense}), the updating of $\y_k$ is
\[
\y_{k+1} = \y_k + \frac{1}{2} \sum_{i\in\T_k} \frac{\tilde{b}_{i} - \langle \tilde{A}_{i*} |D_iA^\dag |\y_{k} \rangle }{\|A_{i*}\|}  \e_i.
\]
So $\y_k$ is at most $(k\kappa_F^2/\kappa^2)$-sparse. Thus the cost in step $k$ is $O(kd\kappa_F^2(\s\q(A) + \q(\b)) /\kappa^2)=\widetilde{O}(k\kappa_F^4(\s\q(A) + \q(\b)) /\kappa^2\epsilon^2)$. After convergence, the total cost is
$\widetilde{O}(T^2\kappa_F^4(\s\q(A) + \q(\b)) /\kappa^2\epsilon^2) =\widetilde{O}(\kappa_F^4\kappa^2(\s\q(A) + \q(\b)) /\epsilon^2)$.
By the estimation in Appendix \ref{app:Estimation of phi for Kaczmarz with averaging},  we have $\phi = O(\kappa_F^2 \kappa^2 + \alpha \kappa^6)$, so the cost to obtain $SQ(\x_T)$ is $O(\phi T\kappa_F^2\s\q(A)/\kappa^2) = \widetilde{O}((\kappa_F^4\kappa^2+\alpha \kappa_F^2\kappa^6)\s\q(A))$.
\end{proof}

\subsection{Solving symmetric positive definite linear systems}

Now suppose $A$ is a symmetric positive definite matrix. Then a simple approach to solve $A\x=\b$ is the randomized coordinate descent iteration~\cite{gower2015randomized} 
\be \label{randomized coordinate descent}
\x_{k+1} = \x_k - \frac{\langle A_{r_k*}| \x_k\rangle - b_{r_k}}{A_{r_k,r_k}} \e_{r_k},
\ee
where the probability of choosing $r_k$ is $A_{r_k,r_k}/{\rm Tr}(A)$.
If we apply the above iteration to the normal equation $A^\dag A \x = A^\dag \b$, we will obtain the Kaczmarz method.
As for the iteration (\ref{randomized coordinate descent}), it is not hard to show that 
\[
\E[\|\x_k-\x_*\|^2] \leq \left(1 - \frac{1}{\|A^{+}\| {\rm Tr(A)}} \right)^k \|\x_0-\x_*\|^2.
\]
In fact, assume $A\x_*=\b$, then (\ref{randomized coordinate descent}) implies that
\[
\x_{k+1} - \x_* = \left(I - \frac{\ket{r_k}\bra{r_k} A}{A_{r_k,r_k}} \right)(\x_{k} - \x_*).
\]
Thus
\beas
\E[\|\x_{k+1} - \x_*\|^2] &\leq& \left\| \E \left[I - \frac{\ket{r_k}\bra{r_k} A}{A_{r_k,r_k}} \right] \right\| \E[\|\x_{k} - \x_*\|^2] \\
&=& \left\|I - \frac{A}{{\rm Tr}(A)} \right\|  \E[\|\x_{k} - \x_*\|^2]   \\
&=& \left(1 - \frac{1}{\|A^{+}\| {\rm Tr(A)}} \right)\E[\|\x_{k} - \x_*\|^2].
\eeas

Similar to the idea introduced above to handle Kaczmarz method, we can use the randomized coordinate descent method to solve dense linear systems $A\x=\b$ such that $A$ is symmetric positive definite. The only difference is that $\|A\|_F^2$ becomes ${\rm Tr}(A)$, and $\|A^{+}\|^2$ becomes $\|A^{+}\|$. Moreover, the randomized coordinate descent (\ref{randomized coordinate descent}) can also be parallelized into
\be
\x_{k+1} = \x_k - \frac{\alpha}{2q} \sum_{i\in \T_k} \frac{\langle A_{i*}| \x_k\rangle - b_{i}}{A_{i,i}} \e_{i}.
\ee
If we take $\alpha=q={\rm Tr}(A)/\|A\|$, then it converges after $O(\kappa \log(1/\epsilon))$ iterations. The proof is the same as that of (\ref{Kaczmarz method in average}). Therefore, by a similar argument to the proof of Theorem \ref{thm:improvements}, we have

\begin{thm}
\label{thm:spd}
Let $A$ be symmetric positive definite.
Assume that $\x_* = \arg\min \|A\x - \b\|$.
Given $SQ(A), Q(\b)$, there is an algorithm that returns $SQ(\x)$ such that $\|\x - \x_*\| \leq \epsilon \|\x_*\|$ with probability at least 0.99 and
\be
\s\q(\x) = \widetilde{O}\left(\frac{{\rm Tr}(A)^2\|A^{+}\|^2\kappa (\s\q(A)+\q(\b))}{\epsilon^2} \right).
\ee
\end{thm}

The above result is actually more general than Theorem \ref{thm:improvements}. If we consider the normal equation $A^\dag A \x = A^\dag \b$, Theorem \ref{thm:spd} implies Theorem \ref{thm:improvements}.

\section{Discussion}
\label{section:Discussion for solving row sparse linear systems}

\subsection{Solving row sparse linear systems}

When $A$ is row sparse, the inner product in the Kaczmarz method (\ref{Randomized Kaczmarz algorithm}) is not expensive to calculate. As a result, we can directly apply the Kaczmarz method to solve the linear system $A\x=\b$. The result is summarized below, and the proof is straightforward. 

\begin{thm}
\label{thm:sparse case}
Assume that $A$ has row sparsity $s$. Let $\x_* = \arg\min \|A\x - \b\|$.
Given $SQ(A), SQ(\b)$, then there is an algorithm that returns an $O(s\kappa_F^2 \log(1/\epsilon))$-sparse vector $\x$ such that $\|\x - \x_*\| \leq \epsilon \|\x_*\|$ with probability at least 0.99 in time
\be
O(s\kappa_F^2 \log (1/\epsilon) ).
\ee
\end{thm}

\begin{proof}
In the Kaczmarz iteration (\ref{Randomized Kaczmarz algorithm}), since $A$ has row sparsity $s$, it follows that the calculation of the inner product $\langle A_{r_k*}|\x_k\rangle$ costs $O(s)$ operations. This is also the cost of the $k$-th step of iteration. After convergence, the total cost is $O(Ts)$, where $T=O(\kappa_F^2 \log(1/\epsilon))$. 
As for the sparsity of $\x_k$, we can choose the initial approximation $\x_0=0$, then $\x_k$ is at most $sk$-sparse.
\end{proof}

Similarly, based on the randomized coordinate descent method (\ref{randomized coordinate descent}), we have

\begin{thm}
\label{thm:sparse case2}
Suppose that $A$ is symmetric positive definite and $s$-sparse. Let $\x_* = \arg\min \|A\x - \b\|$.
Given $SQ(A), SQ(\b)$, then there is an algorithm that returns an $O({\rm Tr}(A) \|A^{+}\| \log (1/\epsilon) )$-sparse vector $\x$ such that $\|\x - \x_*\| \leq \epsilon \|\x_*\|$ with probability at least 0.99 in time
\be
O(s{\rm Tr}(A) \|A^{+}\| \log (1/\epsilon) ).
\ee
\end{thm}

Differently from  quantum and  quantum-inspired algorithms, the outputs of the above two randomized algorithms are vectors. 
It is known that if $A$ has row and column sparsity $s$ and given in the sparse access input model, the currently best known quantum algorithm for matrix inversion costs $O(s\|A\|_{\max} \|A^{+}\| \log^2(1/\epsilon))$, where $\|A\|_{\max}$ refers to the maximal entry in the sense of absolute value \cite{chakraborty_et_al:LIPIcs:2019:10609}. 
The comparison between this quantum algorithm and the randomized algorithm stated in Theorem \ref{thm:sparse case} may be not fair enough because they are building on different models -- the quantum algorithm uses the sparse access model, while the later one assumes access to $SQ(A)$. However, if we only concentrate on the final complexity, the quantum linear solver achieves large speedups if $\|A\|_F \gg \|A\|_{\max}$. This can often happen, especially when the rank of $A$ is large (e.g.\ $A$ is a sparse unitary).
However, when the rank of $A$ is polylog in the dimension, then $\|A\|_F=\widetilde{O}(s\|A\|_{\max})$ if $A$ is row and column sparse.  In this case, the quantum computer only achieves a quadratic speedup.
For a row sparse but column dense matrix $A$, if it is stored in the QRAM data structure, the complexity to solve $A\x=\b$ using a quantum computer is $O(\kappa_F\log^2(1/\epsilon))$~\cite{chakraborty_et_al:LIPIcs:2019:10609}. This only gives a quadratic speedup in terms of $\kappa_F$ over Theorem \ref{thm:sparse case}.

By Theorem \ref{thm:sparse case2}, when $A$ is symmetric positive definite, the randomized coordinate descent method finds an approximate solution in time $O({\rm Tr}(A) \|A^{+}\| \log(1/\epsilon))$. For simplicity, let us consider the special case that $A$ is a sparse density operator, then ${\rm Tr(A)} = 1$. Under certain conditions \cite{gilyen2019quantum}, we can efficiently find a unitary such that $A$ is a sub-block. This means that the complexity of the quantum linear solver for this particular linear system is $O(\|A\|\|A^{+}\|\log^2(1/\epsilon))$, while the algorithm based on Theorem \ref{thm:sparse case2} costs $O(\|A^{+}\| \log(1/\epsilon))$. So if $\|A\|=\Theta(1)$, it is possible that there is no quantum speedup based on the current quantum linear solvers.

In conclusion, quantum linear solvers can build on many different models, e.g. sparse-access input model, QRAM model, etc. Whenever we have an efficient block-encoding, quantum computers can solve the corresponding linear system efficiently. In comparison,  so far there is only one model for quantum-inspired algorithms. The above two theorems indicate that to explore large quantum speedups for sparse problems, it is better to focus on other models rather than the QRAM model. This also leads to the problem of exploring other models for quantum-inspired algorithms.

\subsection{Reducing the dependence on $\|A\|_F$}
\label{sec:reduceaf}

The Frobenius norm can be large for high-rank matrices. For example, when we use the discretization method (e.g.\ Euler method) to solve a linear system of ordinary differential equations, the obtained linear system contains the identity as a sub-matrix. So the Frobenius norm can be as large as the dimension of the problem. The performance of quantum-inspired algorithms is highly affected by the Frobenius norm, so currently they are mainly effective for low-rank systems. To enlarge the application scope, it is necessary to consider the problem of how to reduce or even remove the the dependence on the Frobenius norm in the complexity.

To the best of our knowledge, there are no generalizations of the Kaczmarz method such that the total cost is independent of $\|A\|_F$. However, it can be generalized such that $\|\E[\x_k-\x_*]\| \leq (1-\kappa^{-r})^k \|\x_0-\x_*\|$ and each step of iteration is not expensive, where $r\in\{1,2\}$. For example, a simple generalization is 
\[
\x_{k+1} = \x_k + \frac{\|A\|_F^2}{\|A\|^2} (\tilde{b}_{r_k} - \langle \widetilde{A}_{r_k*}|\x_k\rangle) \widetilde{A}_{r_k*}.
\]
Generally, the relaxation parameter ${\|A\|_F^2}/{\|A\|^2} \gg 2$, so
this iterative scheme is not convergent,\footnote{To make sure that it is convergent, the  relaxation parameter should be chosen from $(0,2)$. Usually, the optimal choice is 1.} that is we usually do not have $\E[\|\x_k-\x_*\|^2]\| \leq \rho^k \|\x_0-\x_*\|^2$ for some $\rho<1$. However, we do have $\|\E[\x_k-\x_*]\| \leq (1-\kappa^{-2})^k \|\x_0-\x_*\|$. The proof is straightforward. If we set $T = O(\kappa^2 \log(1/\epsilon))$, then $\|\E[\x_T-\x_*]\| \leq \epsilon \|\x_0-\x_*\|$. There are also some other generalizations \cite{necoara2019faster,richtarik2020stochastic,loizou2020momentum} such that $T=O(\kappa \log(1/\epsilon))$, which is optimal for this kind of iterative schemes. One idea to apply this result is to generate $L$ samples $\x_{T1},\ldots,\x_{TL}$ such that $\| \E[\x_T] - \frac{1}{L} \sum_{l=1}^L \x_{Tl}\| \leq \epsilon \|\E[\x_T]\|$. The matrix Bernstein inequality might be helpful to determine $L$. However, to apply this method, we have to find tight bounds for $\|\x_T-\x_*\|$ and $\E[\|\x_T-\x_*\|^2]$, which seems not that obvious.



\section{Outlook}

In this paper, for solving linear systems, we reduced the separation between quantum and quantum-inspired algorithms from $\kappa_F:\kappa_F^6\kappa^6/\epsilon^4$ to $\kappa_F:\kappa_F^4\kappa^2/\epsilon^2$. It is interesting to consider whether the parameters of our algorithm can be improved further.
For the class of iterative schemes including the Kaczmarz iteration, the optimal convergence rate equals the condition number $\kappa$ (as opposed to $\kappa_F^2$ as in the Kaczmarz method), which is achieved by the conjugate gradient method. However, there is currently no such result for the Kaczmarz iterative schemes themselves.
As discussed in Section \ref{sec:reduceaf}, in \cite{necoara2019faster,richtarik2020stochastic,loizou2020momentum}, it was  proved that there are generalizations of the Kaczmarz method so that $\|\E[\x_T-\x_*]\| \leq (1-\kappa^{-1})^T \|\x_0-\x_*\|$. This convergence result is a little weak but may be helpful to reduce the separation in terms of $\kappa_F$ further.



\section*{Acknowledgement}

We  would like to thank Ryan Mann for useful discussions.
This paper was supported by the QuantERA ERA-NET Cofund in Quantum Technologies implemented within the European
Union's Horizon 2020 Programme (QuantAlgo project), and EPSRC grants EP/L021005/1 and EP/R043957/1. This project has received funding from the European Research Council (ERC) under the European Union’s Horizon 2020 research and innovation programme (grant agreement No.\ 817581).
No new data were created during this study.


\begin{appendices}

\section{Estimation of $\phi$}
\label{app for non-sparse matrices}

Using the notation (\ref{notation of mu}), we can decompose the updating rule (\ref{updating for y'-matrix form})  of $\y$ as follows
\be  \label{app:eq1}
\y_{k+1} = \y_k + \z_k + \z_k',
\ee
where 
\beas
\z_k = \frac{\tilde{b}_{r_k} - \langle \tilde{A}_{r_k*} |  A^\dag| \y_{k} \rangle }{\|A_{r_k*}\|}  \e_{r_k}, \quad
\z_k' = \frac{\mu_k}{\|A_{r_k*}\|} \e_{r_k}.
\eeas

Denote
\be
Z = \|A\x_* - \b\| = \min_{\x} \|A\x - \b\|, \quad
\Lambda = {\rm diag} (\|A_{i*}\|^2: i \in[m]).
\ee
In the following, for any two vectors $\a,\b$, we define $\langle \a|\b\rangle_{\Lambda} = \langle \a|\Lambda|\b\rangle$. To bound $\phi$, it suffices to bound $\|\y_T\|_\Lambda^2$. 
From (\ref{app:eq1}), it is plausible that
$\|\y_{k+1}\|_\Lambda^2 \approx \|\y_{k}\|_\Lambda^2 + \|\z_{k}\|_\Lambda^2 + \|\z_{k}'\|_\Lambda^2$. In the following, we shall prove that in fact this holds up to a constant factor on average. We shall choose the initial vector as 0 for simplicity, i.e., $\y_0=0$.

First, we consider the case  $Z = 0$, that is $\b = A\x_*$. In the following, we first fix $k,r_k$ and compute the mean value over the random variable $D$, then we compute the mean value over $r_k$ by fixing $k$. Finally, we calculate the mean value over the random variable $k$.

From now on, we assume that $d = O((\kappa_F^2/\epsilon^{2}) (\log 1/\epsilon))$ and $T=O(\kappa_F^2\log(1/\epsilon))$.
By Lemma \ref{lem:mean and variance}, 
\beas
\E_D[\|\z_k'\|_\Lambda^2]
&\leq& \frac{\|A\|_F^2 \|\x_k\|^2}{d\min_{j \in [n]} \|A_{*j}\|^2} 
= \frac{\epsilon^2\|\x_k\|^2}{4T} \leq \frac{\epsilon^2(\|\x_k-\x_*\|^2+\|\x_*\|^2)}{2T}, \\
\E_D[\langle \z_k|\z_k'\rangle_\Lambda]
&=& \left( \tilde{b}_{r_k} - \langle \tilde{A}_{r_k*}|\x_k\rangle
\right) \E_D[\mu_k] = 0, \\
\E_D[\langle \y_k|\z_k'\rangle_\Lambda]
&=& \|A_{r_k*}\| \langle \y_k|\e_{r_k}\rangle  \E_D[\mu_k] = 0.
\eeas
About the norm of $\z_k$, we have
\beas
\E[\|\z_k \|_\Lambda^2] &=& \sum_{r_k = 1}^m \frac{{\|A_{r_k*}\|}^2}{\|A\|_F^2} \frac{(\tilde{b}_{r_k} - \langle \tilde{A}_{r_k*} |  A^\dag| \y_{k} \rangle)^2 }{\|A_{r_k*}\|^2} \|A_{r_k*}\|^2 \\
&=& \frac{1}{\|A\|_F^2} \sum_{r_k = 1}^m  (b_{r_k} - \langle A_{r_k}| \x_{k} \rangle)^2 \\
&\leq& 2\frac{\|\b\|^2 + \|A\|_F^2 \|\x_k\|^2}{\|A\|_F^2}
\leq 2\frac{\|\b\|^2 }{\|A\|_F^2}
+4 \|\x_*\|^2 + 4\|\x_*-\x_k\|^2.
\eeas
As for the inner product between $\y_k$ and $\z_k$, we have the following estimate
\beas
\E[\langle \y_k|\z_k \rangle_\Lambda] &=& \sum_{r_k = 1}^m \frac{{\|A_{r_k*}\|}^2}{\|A\|_F^2} \frac{(\tilde{b}_{r_k} - \langle \tilde{A}_{r_k*} |  A^\dag| \y_{k} \rangle) }{\|A_{r_k*}\|} \langle \y_k|\e_{r_k}\rangle \|A_{r_k*}\|^2 \\
&=& \sum_{r_k = 1}^m \frac{{\|A_{r_k*}\|}^2}{\|A\|_F^2} {(b_{r_k} - \langle A_{r_k*} |  A^\dag| \y_{k} \rangle) } \langle \y_k|\e_{r_k}\rangle  \\
&=& \frac{\langle \b| \y_k\rangle_\Lambda - \|\x_k\|_\Lambda^2}{\|A\|_F^2} 
= \frac{\langle \x_*|A^\dag| \y_k\rangle_\Lambda - \|\x_k\|_\Lambda^2}{\|A\|_F^2} \\
&=& \frac{\langle \x_*| \x_k\rangle_\Lambda - \|\x_k\|_\Lambda^2}{\|A\|_F^2} 
\leq  \|\x_* \| \|\x_k\| + \|\x_k\|^2\\ 
&\leq&  3\|\x_* \|^2 +\|\x_* \|\|\x_k-\x_*\| +2\|\x_k-\x_*\|^2 \\
&\leq&  \frac{7}{2}\|\x_* \|^2 +\frac{5}{2}\|\x_k-\x_*\|^2 .
\eeas
In the above, we used the fact that for any two vectors $\a,\b$ we have
$|\langle \a| \b\rangle_\Lambda| \leq \|\Lambda\|^2 |\langle \a| \b\rangle|$, and $\|\Lambda\|^2 \leq \|A\|_F^2$.

Hence, we have
\beas
\E[\|\y_{k+1}\|_\Lambda^2] 
&=& 
\E[\|\y_{k}\|_\Lambda^2] + \E[\|\z_{k}\|_\Lambda^2] + \E[\|\z_{k}'\|_\Lambda^2] 
+2\E[\langle \y_k |\z_k\rangle_\Lambda] \\
&\leq& \E[\|\y_k \|_\Lambda^2] + 2\frac{\|\b\|^2 }{\|A\|_F^2}
+ (9+ \frac{\epsilon^2}{2T})\E[\|\x_k-\x_*\|^2] + (11+\frac{\epsilon^2}{2T}) \|\x_*\|^2 \\
&\leq& \E[\|\y_k \|_\Lambda^2] + 2\frac{\|\b\|^2 }{\|A\|_F^2}
+ (20 + \frac{\epsilon^2}{T}) \|\x_*\|^2,
\eeas
where we use that $\E[\|\x_k-\x_*\|^2] \leq \|\x_*\|^2$ by Lemma \ref{lem:converge rate}.
Therefore,
\beas
\E[\|\y_T \|_\Lambda^2] \leq T\left(2\frac{\|\b\|^2 }{\|A\|_F^2}
+ (20 + \frac{\epsilon^2}{T}) \|\x_*\|^2\right).
\eeas
This means that, with high probability,
\[
\phi = T \frac{\|\y_T\|_{\Lambda}^2}{\|\x_T\|^2} 
\leq T^2\left(2\frac{\|\b\|^2 }{\|A\|_F^2\|\x_T\|^2}
+ (20 + \frac{\epsilon^2}{T}) \frac{\|\x_*\|^2}{\|\x_T\|^2}\right)
=O(T^2).
\]

When $Z\neq 0$, then $\b = A\x_* + \c$ for some vector $\c$ of norm $Z$ that is not in the column space of $A$. This can happen when $A$ is not full rank. Since $\c$ is independent of $A$, we cannot bound it in terms of $A$. In this case, the only change is $\E[\langle \y_k|\z_k \rangle_\Lambda]$, which is now bounded by
\beas
\E[\langle \y_k|\z_k \rangle_\Lambda] &\leq& \frac{\langle \c|\E[\y_k]\rangle_\Lambda}{\|A\|_F^2}+\frac{7}{2}\|\x_* \|^2 +\frac{5}{2}\|\x_k-\x_*\|^2 \\
&\leq& \frac{\|A\|^2}{\|A\|_F^2}\|\c\| \|\E[\y_k]\| + \frac{7}{2}\|\x_* \|^2 +\frac{5}{2}\|\x_k-\x_*\|^2 \\
&=& \frac{\kappa^2Z}{\kappa_F^2} \|\E[\y_k]\| + \frac{7}{2}\|\x_* \|^2 +\frac{5}{2}\|\x_k-\x_*\|^2.
\eeas
At the end, we obtain
\beas
\E[\|\y_{k+1}\|_\Lambda^2] 
\leq \E[\|\y_k\|_\Lambda^2] + 2\frac{\|\b\|^2 }{\|A\|_F^2} +
 (20 + \frac{\epsilon^2}{T}) \|\x_*\|^2 + \frac{2\kappa^2Z}{\kappa_F^2} \|\E[\y_k]\|.
\eeas
From (\ref{app:eq1}), we know that
\[
\E[\y_{k+1}]
=\left( I - \frac{AA^\dag}{\|A\|_F^2}\right)\E[\y_k] + \frac{\b}{\|A\|_F^2}.
\]
This means
\[
\|\E[\y_k]\| = \left\|\sum_{i=0}^{k-1} \left( I - \frac{AA^\dag}{\|A\|_F^2}\right)^i \frac{\b}{\|A\|_F^2}\right\|
\leq \sum_{i=0}^{k-1} ( 1 - \kappa_F^{-2})^i \frac{\|\b\|}{\|A\|_F^2}
\leq \frac{\kappa_F^2 \|\b\|}{\|A\|_F^2}.
\]
Therefore, 
\beas
\E[\|\y_T\|_\Lambda^2] 
&\leq&
T\left( \frac{2\|\b\|^2  }{\|A\|_F^2}  +
 (20 + \frac{\epsilon^2}{T}) \|\x_*\|^2 +  \frac{2\kappa^2 \|\b\|Z}{\|A\|_F^2} \right) \\
&=&
T\left( \frac{2\|A\x_*\|^2  }{\|A\|_F^2}  +
 (20 + \frac{\epsilon^2}{T}) \|\x_*\|^2 +  \frac{2\kappa^2 \|\b\|Z + 2Z^2}{\|A\|_F^2}\right) .
\eeas
Finally, by Markov's inequality, with a high probability  we have
\[
\phi = T \frac{\|\y_T\|_{\Lambda}^2}{\|\x_T\|^2} = O\left(T^2+ T^2 \frac{\kappa^2 \|\b\|Z + Z^2}{\|A\|_F^2\|\x_*\|^2}  \right) ,
\]
where we used the fact that $\x_T \approx \x_*$ and $\|A\x_*\|^2 / \|A\|_F^2 \leq \|\x_*\|^2$.

Let $Z=\eta \|\b\|$, then  $\|\x_*\|^2 \geq \|A\x_*\|^2/\|A\|^2 = (1-\eta^2)\|\b\|^2/\|A\|^2$. So
\bes
T^2 \frac{\kappa^2 \|\b\|Z}{\|A\|_{\F}^2\|\x_*\|^2}
\leq 
\frac{\eta}{1-\eta^2} \kappa_F^2 \kappa^4.
\ees
This means 
\[
\phi = O\left(\kappa_F^4 + \frac{\eta}{1-\eta^2} \kappa_F^2 \kappa^4 \right).
\]


\section{Estimation of the convergence rate}
\label{app for the convergence rate}

For simplicity, we assume that $A\x_* = \b$. Following the analysis of \cite{moorman2020randomized}, the convergence rate does not change too much if the LSP is not consistent. 
By the updating formula (\ref{kaczmarz average:dense}), we have
\beas
\x_{k+1} - \x_* &=& \x_k - \x_* + \frac{1}{2} \sum_{i\in\T_k} (\tilde{b}_{i} - \langle \tilde{A}_{i*} |\x_{k}\rangle )  \ket{\tilde{A}_{i*}}
+ \frac{1}{2} \sum_{i\in\T_k} \langle\tilde{A}_{i*} | I - D_i |\x_{k}\rangle \ket{\tilde{A}_{i*}} \\
&=& \left( I - \frac{1}{2} \sum_{i\in\T_k}  \ket{\tilde{A}_{i*}} \bra{\tilde{A}_{i*}}\right) (\x_k - \x_*)
+ \frac{1}{2} \sum_{i\in\T_k} \langle\tilde{A}_{i*} | I - D_i |\x_{k}\rangle \ket{\tilde{A}_{i*}}.
\eeas
Below, we try to bound $\E[\|\x_{k+1} - \x_*\|^2]$. We will follow the notation of (\ref{notation of mu}). But here to avoid any confusion, we denote
\[
\mu_{ik} := \langle\tilde{A}_{i*} | I - D_i |\x_{k} \rangle.
\]
In the following, the result of Lemma \ref{lem:mean and variance} will be used, and $|\T_k| = \|A\|_F^2/\|A\|^2$.

First, we have
\beas
\E_{D}\left[\left\|\sum_{i\in\T_k} \langle\tilde{A}_{i*} | I - D_i |\x_{k} \rangle \ket{\tilde{A}_{i*}}\right\|^2\right] 
&=& \sum_{i,j\in\T_k} \langle \tilde{A}_{i*}|\tilde{A}_{j*}\rangle  \E_{D_i,D_j}\left[ \mu_{ik}\mu_{jk} \right] \\
&=& \sum_{i\in\T_k}  \E_{D_i} [ \mu_{ik}^2 ]
+\sum_{i\neq j} \langle \tilde{A}_{i*}|\tilde{A}_{j*}\rangle  \E_{D_i}\left[ \mu_{ik}\right] \E_{D_j}[\mu_{jk}] \\
&\leq& \frac{1}{d}
\sum_{i\in\T_k} \sum_{j=1}^n 
\tilde{A}_{i,j}^2 x_{k,j}^2 \frac{\|A\|_F^2}{\|A_{*j}\|^2}.
\eeas
By Lemma \ref{lem:mean and variance}, we have
\beas
&& \E_D\left[\left\langle \left( I - \frac{1}{2} \sum_{i\in\T_k}  \ket{\tilde{A}_{i*}} \bra{\tilde{A}_{i*}}\right) (\x_k - \x_*) \Bigg| \sum_{i\in\T_k} \mu_{ik} \ket{\tilde{A}_{i*}} \right\rangle \right] \\
&=&  \left\langle \left( I - \frac{1}{2}\sum_{i\in\T_k}  \ket{\tilde{A}_{i*}} \bra{\tilde{A}_{i*}}\right) (\x_k - \x_*) \Bigg| \sum_{i\in\T_k} \E_D[\mu_{ik}] \ket{\tilde{A}_{i*}} \right\rangle = 0.
\eeas
Hence, after computing the mean value over $D$, we have
\beas
\|\x_{k+1} - \x_*\|^2
\leq \left\| \left( I - \frac{1}{2} \sum_{i\in\T_k}  \ket{\tilde{A}_{i*}} \bra{\tilde{A}_{i*}}\right) (\x_k - \x_*) \right\|^2
+  \frac{1}{4d}
\sum_{i\in\T_k} \sum_{j=1}^n 
\tilde{A}_{i,j}^2 x_{k,j}^2 \frac{\|A\|_F^2}{\|A_{*j}\|^2}.
\eeas

For the first term, we now compute its mean value over the random variable $\T_k$
\beas
&& \E\left[\left\| \left( I -\frac{1}{2} \sum_{i\in\T_k}  \ket{\tilde{A}_{i*}} \bra{\tilde{A}_{i*}}\right) (\x_k - \x_*) \right\|^2\right] \\
&=& \langle \x_k - \x_* | 
\E\left[
\left( I -\frac{1}{2} \sum_{i\in\T_k}  \ket{\tilde{A}_{i*}} \bra{\tilde{A}_{i*}}\right)^2
\right]
| \x_k - \x_*\rangle.
\eeas
And it can be shown that
\beas
&& \E\left[
\left( I -\frac{1}{2} \sum_{i\in\T_k}  \ket{\tilde{A}_{i*}} \bra{\tilde{A}_{i*}}\right)^2
\right] \\
&=& \E\left[I - \frac{3}{4}\sum_{i\in\T_k}  \ket{\tilde{A}_{i*}} \bra{\tilde{A}_{i*}} + \frac{1}{4} \sum_{i,j\in\T_k, i\neq j} \ket{\tilde{A}_{i*}} \langle \tilde{A}_{i*}|\tilde{A}_{j*}\rangle \bra{\tilde{A}_{j*}}\right] \\
&=& I -  \frac{3q}{4}\frac{A^\dag A}{\|A\|_F^2} +  \frac{1}{4}  (q^2-q) \left( \frac{A^\dag A}{\|A\|_F^2}  \right)^2 \\
&\preceq& \left(1-\frac{3q}{4} \frac{\sigma_{\min}(A^\dag A)}{\|A\|_F^2}+  \frac{1}{4}  (q^2-q) \left(\frac{\sigma_{\min}(A^\dag A)}{\|A\|_F^2}  \right)^2 \right) I \\
&\preceq& \left(1-\frac{1}{2\kappa^2} \right) I.
\eeas
In the last step, we used the result  $q=\|A\|_F^2/\|A\|^2$ so that the second term is $3/4\kappa^2$ and the third term is less than $1/4\kappa^4 \leq 1/4\kappa^2$.

Therefore, 
\beas
\E[\|\x_{k+1} - \x_*\|^2]
&\leq& (1-\frac{1}{2\kappa^2}) \E[\|\x_{k} - \x_*\|^2]
+  \frac{1}{4d}
\E\left[\sum_{i\in\T_k} \sum_{j=1}^n 
\tilde{A}_{i,j}^2 x_{k,j}^2 \frac{\|A\|_F^2}{\|A_{*j}\|^2} \right] \\
&=& (1-\frac{1}{2\kappa^2}) \E[\|\x_{k} - \x_*\|^2]
+  \frac{\|A\|_F^2}{4d\|A\|^2}
\E[\|\x_k\|^2].
\eeas
Now set $T = O(\kappa^2 \log(2/\epsilon^2))$, and
\be \label{app:value of d}
d = \frac{T \|A\|_F^2}{\|A\|^2 \epsilon^2} 
= \frac{\kappa_F^2\log(2/\epsilon^2)}{\epsilon^2}.
\ee
Then
\beas
\E[\|\x_{k+1} - \x_*\|^2]
&\leq& (1-\frac{1}{2\kappa^2}) \E[\|\x_{k} - \x_*\|^2]
+ \frac{\epsilon^2}{4T} \E[\|\x_k\|^2] \\
&\leq& (1-\frac{1}{2\kappa^2}+ \frac{\epsilon^2}{2T}) \E[\|\x_{k} - \x_*\|^2]
+ \frac{\epsilon^2}{2T}  \|\x_*\|^2.
\eeas
Finally, we obtain
\beas
\E[\|\x_T - \x_*\|^2]
\lesssim (1-\frac{1}{2\kappa^2} + \frac{\epsilon^2}{2T})^T \|\x_{0} - \x_*\|^2
+  \frac{\epsilon^2}{2}  \|\x_*\|^2 
\leq \epsilon^2 \|\x_*\|^2 .
\eeas

\section{Estimation of $\phi$ for Kaczmarz method with averaging}
\label{app:Estimation of phi for Kaczmarz with averaging}

The calculation here is similar to that in Appendix  \ref{app for non-sparse matrices}.
For simplicity, denote
\bes
\y_{k+1} = \y_k + \frac{1}{2} \w_k +  \frac{1}{2}  \w_k ',
\ees
where
\beas
\w_k := \sum_{i \in \T_k} \frac{b_{i} - \langle A_{i*}| \x_k \rangle   }{\|A_{i*}\|^2} \e_{i*},
\quad
\w_k ' :=  \sum_{i\in\T_k} \frac{\mu_{ik}}{\|A_{i*}\|} \e_{i*}.
\eeas

In this section, $d$ is given in the formula (\ref{app:value of d}).
From a similar estimation in Appendix \ref{app for non-sparse matrices}, we have $\E_D[\langle \w_k|\w_k'\rangle_\Lambda]
=\E_D[\langle \y_k|\w_k'\rangle_\Lambda]
 = 0$ and
\beas
\E_{r_k}\E_D[\|\w_k'\|_\Lambda^2] 
\leq \frac{\|A\|_F^2}{\|A\|^2} \frac{1}{d}  \sum_{j=1}^n 
\E_{r_k}[\tilde{A}_{r_k,j}^2] x_{k,j}^2 \frac{\|A\|_F^2}{\|A_{*j}\|^2} 
= \frac{\epsilon^2\|\x_k\|^2}{4T} \leq \frac{\epsilon^2(\|\x_k-\x_*\|^2+\|\x_*\|^2)}{2T}.
\eeas
As for $\|\w_k\|_\Lambda^2$, we still have
\beas
\E[\|\w_k\|_\Lambda^2]
&=&\E\left[ \sum_{i,j \in \T_k} \frac{b_{i} - \langle A_{i*}| \x_k \rangle   }{\|A_{i*}\|^2} \frac{b_{j} - \langle A_{j*}| \x_k \rangle   }{\|A_{j*}\|^2} \langle \e_{i*}|\e_{j*} \rangle \|A_{i*}\|^2 \right] \\
&=& \frac{\|A\|_F^2}{\|A\|^2} \E\left[\frac{(b_{i} - \langle A_{i*}| \x_k \rangle)^2 }{\|A_{i*}\|^2}\right] \\
&=& \frac{\|\b - A\x_k\|^2}{\|A\|^2} \\
&\leq& \frac{2\|\b\|^2}{\|A\|^2} + 2\|\x_k\|^2.
\eeas
By the estimation of $\E[\langle \y_k|\z_k\rangle_{\Lambda}]$ in Appendix \ref{app for non-sparse matrices} and noting that $\|\Lambda\|\leq \|A\|$, we also have
\beas
\E[\langle \y_k|\w_k\rangle_{\Lambda}]
= \frac{\|A\|_F^2}{\|A\|^2} \E[\langle\y_k|\frac{b_{i} - \langle A_{i*}| \x_k \rangle   }{\|A_{i*}\|^2} \e_{i*}\rangle]
\leq \frac{7}{2}\|\x_* \|^2 +\frac{5}{2}\|\x_k-\x_*\|^2 .
\eeas
All the estimations above do not change. The constant 1/2 in the decomposition of  $\y_{k+1}$ does not affect the upper bound, so when $Z=0$ we have
\beas
\E[\|\y_T \|_\Lambda^2] \leq T\left(2\frac{\|\b\|^2 }{\|A\|_F^2}
+ (20 + \frac{\epsilon^2}{T}) \|\x_*\|^2\right),
\eeas
where $T=\widetilde{O}(\kappa^2)$.
Therefore
\[
\phi = T\frac{\kappa_F^2}{\kappa^2} \frac{\|\y_T\|_{\Lambda}^2}{\|\x_T\|^2} 
\leq T^2\frac{\kappa_F^2}{\kappa^2}\left(2\frac{\|\b\|^2 }{\|A\|_F^2\|\x_T\|^2}
+ (20 + \frac{\epsilon^2}{T}) \frac{\|\x_*\|^2}{\|\x_T\|^2}\right)
=O(\kappa_F^2\kappa^2).
\]
When $Z\neq 0$, we similarly have 
\[
\phi = O\left(\kappa_F^2 \kappa^2 + \frac{\eta}{1-\eta^2} \kappa^6\right).
\]

\end{appendices}

\bibliographystyle{plain}
\bibliography{main}

\begin{thebibliography}{10}

\bibitem{arrazola2019quantum}
Juan~Miguel Arrazola, Alain Delgado, Bhaskar~Roy Bardhan, and Seth Lloyd.
\newblock Quantum-inspired algorithms in practice.
\newblock {\em Quantum}, 4:307, 2020.
\newblock {\tt arXiv:1905.10415}.

\bibitem{biamonte2017quantum}
Jacob Biamonte, Peter Wittek, Nicola Pancotti, Patrick Rebentrost, Nathan
  Wiebe, and Seth Lloyd.
\newblock Quantum machine learning.
\newblock {\em Nature}, 549(7671):195--202, 2017.

\bibitem{censor1997parallel}
Yair Censor, Stavros~Andrea Zenios, et~al.
\newblock {\em Parallel optimization: Theory, algorithms, and applications}.
\newblock Oxford University Press on Demand, 1997.

\bibitem{chakraborty_et_al:LIPIcs:2019:10609}
Shantanav Chakraborty, Andr{\'a}s Gily{\'e}n, and Stacey Jeffery.
\newblock {The Power of Block-Encoded Matrix Powers: Improved Regression
  Techniques via Faster Hamiltonian Simulation}.
\newblock In {\em 46th International Colloquium on Automata, Languages, and
  Programming (ICALP 2019)}, pages 33:1--33:14, 2019.
\newblock {\tt arXiv:1804.01973}.

\bibitem{chepurko2020quantum}
Nadiia Chepurko, Kenneth~L Clarkson, Lior Horesh, and David~P Woodruff.
\newblock Quantum-inspired algorithms from randomized numerical linear algebra.
\newblock {\em arXiv preprint arXiv:2011.04125}, 2020.

\bibitem{chia2020sampling}
Nai-Hui Chia, Andr{\'a}s Gily{\'e}n, Tongyang Li, Han-Hsuan Lin, Ewin Tang, and
  Chunhao Wang.
\newblock Sampling-based sublinear low-rank matrix arithmetic framework for
  dequantizing quantum machine learning.
\newblock In {\em Proceedings of the 52nd Annual ACM SIGACT Symposium on Theory
  of Computing (STOC 2020)}, pages 387--400, 2020.
\newblock {\tt arXiv:1910.06151}.

\bibitem{chia2019quantum}
Nai-Hui Chia, Tongyang Li, Han-Hsuan Lin, and Chunhao Wang.
\newblock Quantum-inspired classical algorithms for singular value
  transformation.
\newblock In {\em 45th International Symposium on Mathematical Foundations of
  Computer Science (MFCS 2020)}, pages 23:1--23:15. Schloss
  Dagstuhl-Leibniz-Zentrum f{\"u}r Informatik, 2020.
\newblock {\tt arXiv:1901.03254}.

\bibitem{chia2018quantum2}
Nai-Hui Chia, Han-Hsuan Lin, and Chunhao Wang.
\newblock Quantum-inspired sublinear classical algorithms for solving low-rank
  linear systems.
\newblock {\em arXiv preprint arXiv:1811.04852}, 2018.

\bibitem{childs2017quantum}
Andrew~M Childs, Robin Kothari, and Rolando~D Somma.
\newblock Quantum algorithm for systems of linear equations with exponentially
  improved dependence on precision.
\newblock {\em SIAM Journal on Computing}, 46(6):1920--1950, 2017.

\bibitem{drineas2006fast}
Petros Drineas, Ravi Kannan, and Michael~W Mahoney.
\newblock {Fast Monte Carlo algorithms for matrices II: Computing a low-rank
  approximation to a matrix}.
\newblock {\em SIAM Journal on computing}, 36(1):158--183, 2006.

\bibitem{gilyen2018quantum}
Andr{\'a}s Gily{\'e}n, Seth Lloyd, and Ewin Tang.
\newblock Quantum-inspired low-rank stochastic regression with logarithmic
  dependence on the dimension.
\newblock 2018.
\newblock {\tt arXiv:1811.04909}.

\bibitem{gilyen2020improved}
Andr{\'a}s Gily{\'e}n, Zhao Song, and Ewin Tang.
\newblock An improved quantum-inspired algorithm for linear regression.
\newblock 2020.
\newblock {\tt arXiv:2009.07268}.

\bibitem{gilyen2019quantum}
Andr{\'a}s Gily{\'e}n, Yuan Su, Guang~Hao Low, and Nathan Wiebe.
\newblock Quantum singular value transformation and beyond: exponential
  improvements for quantum matrix arithmetics.
\newblock In {\em Proceedings of the 51st Annual ACM SIGACT Symposium on Theory
  of Computing}, pages 193--204, 2019.

\bibitem{golub13}
Gene~H. Golub and Charles~F. van Loan.
\newblock {\em Matrix Computations}.
\newblock The Johns Hopkins University Press, Baltimore, fourth edition, 2013.

\bibitem{gordon1970algebraic}
Richard Gordon, Robert Bender, and Gabor~T Herman.
\newblock Algebraic reconstruction techniques (art) for three-dimensional
  electron microscopy and x-ray photography.
\newblock {\em Journal of theoretical Biology}, 29(3):471--481, 1970.

\bibitem{gower2015randomized}
Robert~M Gower and Peter Richt{\'a}rik.
\newblock Randomized iterative methods for linear systems.
\newblock {\em SIAM Journal on Matrix Analysis and Applications},
  36(4):1660--1690, 2015.

\bibitem{harrow2009quantum}
Aram~W Harrow, Avinatan Hassidim, and Seth Lloyd.
\newblock Quantum algorithm for linear systems of equations.
\newblock {\em Phys Rev Lett}, 103(15):150502, 2009.
\newblock {\tt arXiv:0811.3171}.

\bibitem{jethwani2019quantum}
Dhawal Jethwani, Fran{\c{c}}ois Le~Gall, and Sanjay~K Singh.
\newblock Quantum-inspired classical algorithms for singular value
  transformation.
\newblock In {\em 45th International Symposium on Mathematical Foundations of
  Computer Science (MFCS 2020)}, page 53:1–53:14. Schloss
  Dagstuhl-Leibniz-Zentrum f{\"u}r Informatik, 2020.
\newblock {\tt arXiv:1910.05699}.

\bibitem{karczmarz1937angenaherte}
S~Karczmarz.
\newblock Angenaherte auflosung von systemen linearer glei-chungen.
\newblock {\em Bull. Int. Acad. Pol. Sic. Let., Cl. Sci. Math. Nat.}, pages
  355--357, 1937.

\bibitem{kerenidis2019qmeans}
Iordanis Kerenidis, Jonas Landman, Alessandro Luongo, and Anupam Prakash.
\newblock q-means: A quantum algorithm for unsupervised machine learning.
\newblock In {\em Advances in Neural Information Processing Systems}, pages
  4134--4144, 2019.

\bibitem{kerenidis2017quantum}
Iordanis Kerenidis and Anupam Prakash.
\newblock {Quantum Recommendation Systems}.
\newblock In {\em 8th Innovations in Theoretical Computer Science Conference
  (ITCS 2017)}, pages 49:1--49:21, 2017.
\newblock {\tt arXiv:1603.08675}.

\bibitem{kerenidis2020quantum}
Iordanis Kerenidis and Anupam Prakash.
\newblock A quantum interior point method for {LPs and SDPs}.
\newblock {\em ACM Transactions on Quantum Computing}, 1(1):1--32, 2020.

\bibitem{leventhal2010randomized}
Dennis Leventhal and Adrian~S Lewis.
\newblock Randomized methods for linear constraints: convergence rates and
  conditioning.
\newblock {\em Mathematics of Operations Research}, 35(3):641--654, 2010.

\bibitem{lin2020optimal}
Lin Lin and Yu~Tong.
\newblock Optimal polynomial based quantum eigenstate filtering with
  application to solving quantum linear systems.
\newblock {\em Quantum}, 4:361, 2020.

\bibitem{lloyd2014quantum}
Seth Lloyd, Masoud Mohseni, and Patrick Rebentrost.
\newblock Quantum principal component analysis.
\newblock {\em Nature Physics}, 10(9):631--633, 2014.
\newblock {\tt arXiv:1307.0401}.

\bibitem{loizou2020momentum}
Nicolas Loizou and Peter Richt{\'a}rik.
\newblock Momentum and stochastic momentum for stochastic gradient, newton,
  proximal point and subspace descent methods.
\newblock {\em Computational Optimization and Applications}, 77(3):653--710,
  2020.

\bibitem{moorman2020randomized}
Jacob~D Moorman, Thomas~K Tu, Denali Molitor, and Deanna Needell.
\newblock {Randomized Kaczmarz with averaging}.
\newblock 2020.
\newblock {\tt arXiv:2002.04126}.

\bibitem{necoara2019faster}
Ion Necoara.
\newblock {Faster randomized block Kaczmarz algorithms}.
\newblock {\em SIAM J Matrix Anal Appl}, 40(4):1425--1452, 2019.
\newblock {\tt arXiv:1902.09946}.

\bibitem{nemirovski2009robust}
Arkadi Nemirovski, Anatoli Juditsky, Guanghui Lan, and Alexander Shapiro.
\newblock Robust stochastic approximation approach to stochastic programming.
\newblock {\em SIAM J Optim}, 19(4):1574--1609, 2009.

\bibitem{nesterov2012efficiency}
Yu~Nesterov.
\newblock Efficiency of coordinate descent methods on huge-scale optimization
  problems.
\newblock {\em SIAM J Optim}, 22(2):341--362, 2012.

\bibitem{orsucci2021solving}
Davide Orsucci and Vedran Dunjko.
\newblock On solving classes of positive-definite quantum linear systems with
  quadratically improved runtime in the condition number.
\newblock {\em arXiv preprint arXiv:2101.11868}, 2021.

\bibitem{rebentrost2014quantum}
Patrick Rebentrost, Masoud Mohseni, and Seth Lloyd.
\newblock Quantum support vector machine for big data classification.
\newblock {\em Phys Rev Lett}, 113(13):130503, 2014.
\newblock {\tt arXiv:1307.0471}.

\bibitem{richtarik2020stochastic}
Peter Richt{\'a}rik and Martin Tak{\'a}c.
\newblock Stochastic reformulations of linear systems: algorithms and
  convergence theory.
\newblock {\em SIAM Journal on Matrix Analysis and Applications},
  41(2):487--524, 2020.

\bibitem{shao2020row}
Changpeng Shao and Hua Xiang.
\newblock Row and column iteration methods to solve linear systems on a quantum
  computer.
\newblock {\em Phys Rev A}, 101(2):022322, 2020.
\newblock {\tt arXiv:1905.11686}.

\bibitem{strohmer2009randomized}
Thomas Strohmer and Roman Vershynin.
\newblock A randomized kaczmarz algorithm with exponential convergence.
\newblock {\em J Fourier Anal Appl}, 15(2):262, 2009.
\newblock {\tt arXiv:math/0702226}.

\bibitem{tang2018quantum}
Ewin Tang.
\newblock Quantum-inspired classical algorithms for principal component
  analysis and supervised clustering.
\newblock 2018.
\newblock {\tt arXiv:1811.00414}.

\bibitem{tang2019quantum}
Ewin Tang.
\newblock A quantum-inspired classical algorithm for recommendation systems.
\newblock In {\em Proceedings of the 51st Annual ACM SIGACT Symposium on Theory
  of Computing (STOC 2019)}, pages 217--228, 2019.
\newblock {\tt arXiv:1807.04271}.

\end{thebibliography}

\end{document}